\documentclass[12pt]{article}
\usepackage[left=2cm,top=2cm,right=2cm,bottom=2cm]{geometry}

\usepackage[british]{babel}
\usepackage{amsfonts,amssymb,amsmath,latexsym,hyperref,fancyhdr,amsthm}
\usepackage{graphicx}
\usepackage{epstopdf}
\usepackage{algorithmic}
\usepackage[british]{babel}
\usepackage{enumerate,bbm,xcolor,enumitem,verbatim,stmaryrd}
\usepackage[scr=boondox]{mathalfa}
\usepackage{tabularx}

\newcommand{\ptl}{\partial}
\newcommand{\vph}{\varphi}

\numberwithin{equation}{section}
\pdfoutput=1

\ifpdf
\DeclareGraphicsExtensions{.eps,.pdf,.png,.jpg}
\else
\DeclareGraphicsExtensions{.eps}
\fi

\input{texmacs-macros.texinput}

\newtheorem{theorem}{\bf Theorem}[section]

\definecolor{myred}{RGB}{160,0,0}
\definecolor{mygreen}{RGB}{0,160,0}
\definecolor{myblue}{RGB}{0,0,160}
\hypersetup{
	colorlinks = true,
	linkcolor = {myred},
	citecolor = {mygreen},
	urlcolor  = {myblue},
}

\title{Analytical continuation of two-dimensional wave fields}

\author{Rapha\"{e}l C. Assier$^{*}$ and Andrey V. Shanin$^{\dagger}$\\
	\footnotesize{$^{*}$ Department of Mathematics, University of Manchester, Oxford Road, Manchester, {\rm M13 9PL}, UK}\\
	\footnotesize{$^{\dagger}$ Department of Physics (Acoustics Division), Moscow State University, Leninskie Gory, {\rm 119992}, Moscow, Russia}
}

\begin{document}

	\maketitle

\begin{abstract}
Wave fields obeying the 2D Helmholtz equation on branched surfaces 
(Sommerfeld surfaces) are studied. 
Such surfaces appear naturally as a result of applying the reflection 
method to diffraction problems with straight scatterers bearing ideal 
boundary conditions. This is for example the case for the classical canonical problems of diffraction by a half-line or a segment. In the present work, it is shown that such wave fields admit 
an analytical continuation into the domain of two complex coordinates. 
The branch sets of such continuation are given and studied in detail. For a generic scattering problem, it is shown that the set of all branches of 
the multi-valued analytical continuation of the field  
has a finite basis. Each basis function is expressed explicitly as a Green's integral along so-called double-eight contours. The finite basis property is important in the context 
of coordinate equations, introduced and utilised by the authors previously, as illustrated in this article for the particular case of diffraction by a segment. 
\end{abstract}

\section{Introduction}

In this paper we study 2D diffraction problems for the Helmholtz equation belonging to 
a special class: namely those who can be reformulated as problems of propagation 
on branched surfaces with finitely many sheets. At least two classical canonical diffraction problems belong to this 
class: the Sommerfeld problem of diffraction by a half-line with ideal boundary conditions, 
and the problem of diffraction by an ideal segment. There are also some other important problems belonging to this class, they are listed in Appendix~\ref{app:generalclass}. 
The branched surface for such a problem is referred to as a {\em Sommerfeld surface}
and denoted by~$S$. This surface has several sheets over the real Cartesian plane 
$(x_1 , x_2)$, and these sheets are connected at several branch points.

The solution of the corresponding problem is denoted by $u(x_1, x_2)$ and is assumed to 
be known on $S$.
We consider the possibility to continue the solution $u$ 
into the complex domain of the coordinates $(x_1 ,x_2) \in \mathbb{C}^2$. 
Namely, we are looking for a function $u_{\rm c}(x_1 ,x_2)$ which is 
analytical almost everywhere, obeys the complex Helmholtz equation, and is
equal to $u(x_1 , x_2)$ for real $(x_1 , x_2)$.

It is shown in the paper that such a continuation can be constructed using Green's third identity. The integration contours used involve rather complicated 
loops drawn on $S$, referred to as double-eight or Pochhammer contours. The integrand comprises the function $u$ on $S$, a complexified kernel and their first derivatives. 

The analytical continuation $u_{\rm c}$ is a branched (i.e.\ multivalued) function in $\mathbb{C}^2$. Its 
branch set $T$ is the union of the complex characteristics of the Helmholtz equation passing through the branch points 
of~$S$. At each point of $\mathbb{C}^2$ not belonging to the branch set $T$ 
there exists an infinite number of branches of $u_{\rm c}$, i.e. infinitely many possible values of $u_c$ in the neighbourhood of this point. However, we prove here that these branches have a \textit{finite basis}, 
such that any branch can be expressed as a linear combination of a finite number of \textit{basis functions} with 
integer coefficients. 

Such property of the analytical continuation is an important property of the initial (real) diffraction 
problem. Namely, it indicates that one can build the so-called {\em coordinate equations\/},  
which are ``multidimensional ordinary differential equations''
\cite{Shanin2002a,Shanin2008,Shanin2008a}. Thus, a PDE becomes effectively solved as a finite set of ODEs. 

To emphasise the non-triviality of the statements proven in this paper
and the importance of studying the analytical continuation of the field, 
we can say that we could not generalise the results to the case of 3D diffraction problems yet. It is possible to build a
3D analog of Sommerfeld surfaces (for example, it can be done for the ideal quarter-plane diffraction problem), but at the moment it does not seem possible to show that the analytical continuation 
does possess a finite basis. Thus, a generalisation of the coordinate  
equations seems not to be possible for 3D problems.


 The ideas behind this work were inspired in part by our recent investigations of
applications of multivariable complex analysis to diffraction problems
{\cite{Assier2018a,Assier2019c,Shanin2019SommerfeldTI,Assier2019a,Assier2019b}}, in part by our work on coordinate equations \cite{Shanin2002a,Shanin2008a,Shanin2008,Shanin2005,Assier2012} and in part by the work of Sternin and his co-authors
{\cite{Sternin1994,Kyurkchan1996,Savin2017}}, in which the analytical
continuation of wave fields is considered. 

In  \cite{Sternin1994,Kyurkchan1996,Savin2017}, 
the practical problems of finding minimal 
configurations of sources producing certain fields, or continuation of fields through 
complicated boundaries are solved. 

Different techniques can be used to study the analytical continuation of a wave field. 
We are using Green's theorem as in \cite{Garabedian1960,Sternin1994}. Alternatives include the use of Radon transforms \cite{Sternin1994,Savin2017}, series representations or Schwarz's reflection principle.

The rest of the paper is organised as follows. In Section
\ref{sec:Sommsurfacediffraction} we define the concept of Sommerfeld surface and introduce the notion of diffraction problems on such object. In Section \ref{sec:sec3analcont}, we specify what we mean by the analytical continuation of a wave field $u$ to $\mathbb{C}^2$, and discuss the notion of branching of functions of two complex variables. In Section \ref{sec:intrepofanal}, we give an integral representation that permits to analytically continue $u$ from a point in $S$ to a point in $\mathbb{C}^2$. The obtained function $u_c$ is multi-valued in $\mathbb{C}^2$, and, in Section \ref{sec:analbranchinganal}, we study in detail its branching structure and specify all its possible branches by means of Green's integrals involving double-eight contours. In particular, we show that there exist a finite basis of functions such that any branch of the analytical continuation can be expressed as a linear combination, with integer coefficients, of these basis functions. In Section \ref{sec:strip6}, we apply the general theory developed thus far to the specific problem of diffraction by an ideal strip, showing that, in this case, the number of basis functions can be reduced to four; we describe explicitly and constructively, via some matrix algebra, all possible branches of the analytical continuation. Finally, in Section \ref{sec:coordeq}, still for the strip problem, we show that our results imply the existence of the so-called coordinate equations, effectively reducing the diffraction problem to a set of two multidimensional ODEs.

\section{A diffraction problem on a real Sommerfeld surface}
\label{sec:Sommsurfacediffraction}
Let us start by defining more precisely the concept of Sommerfeld surface. Take $M$ samples of 
the plane $\mathbb{R}^2$ (called \textit{sheets} of the surface), 
each equipped with the Cartesian coordinates $(x_1 , x_2)$. 
Let there exist $N$ affixes of  
branch points $P_1, \dots, P_N \in \mathbb{R}^2$ with coordinates $(X_1^{(j)} , X_2^{(j)})$,
$j = 1,\dots, N$.  Consider
a set of non-intersecting cuts on each sheet, connecting the points 
$P_j$ with each other or with infinity. Finally, let the sides of the cuts be connected (``glued'') to each other 
according to an arbitrary scheme. The connection of the sides should obey the following rules: 
(a) only points having the same coordinates $(x_1 , x_2)$ can be glued to each other;
(b) one can glue a single ``left'' side of a cut to a single ``right'' side  of this cut on another sheet;
(c) a side of a cut should be glued to a side of another cut as a whole. 

Upon allowing spurious cuts, i.e. cuts glued to themselves, it is possible for the set of cuts to be the same on all sheets. The result of assembling the sheets is a Sommerfeld surface
denoted by~$S$.
We assume everywhere that $N$ and $M$ are finite integers.
Two examples of Sommerfeld surfaces are shown in Figure~\ref{fig01}.

The concept of Sommerfeld surfaces is naturally very close to that of Riemann surfaces of 
analytic functions of a complex variable. 
For example, the Sommerfeld surfaces shown in Figure~\ref{fig01} can be 
treated as the Riemann surfaces of the functions $\sqrt{x_1 + i x_2}$ and 
$\sqrt{(x_1 + i x_2)^2 - a^2}$ respectively.  
However, here we prefer to refer to them as Sommerfeld surfaces
since the coordinates $x_1$ and $x_2$ are real on it, and since we 
would like to avoid confusion with the complex context that will 
be developed below. 

\begin{figure}[ht]
\centering
\includegraphics[width=0.8\textwidth]{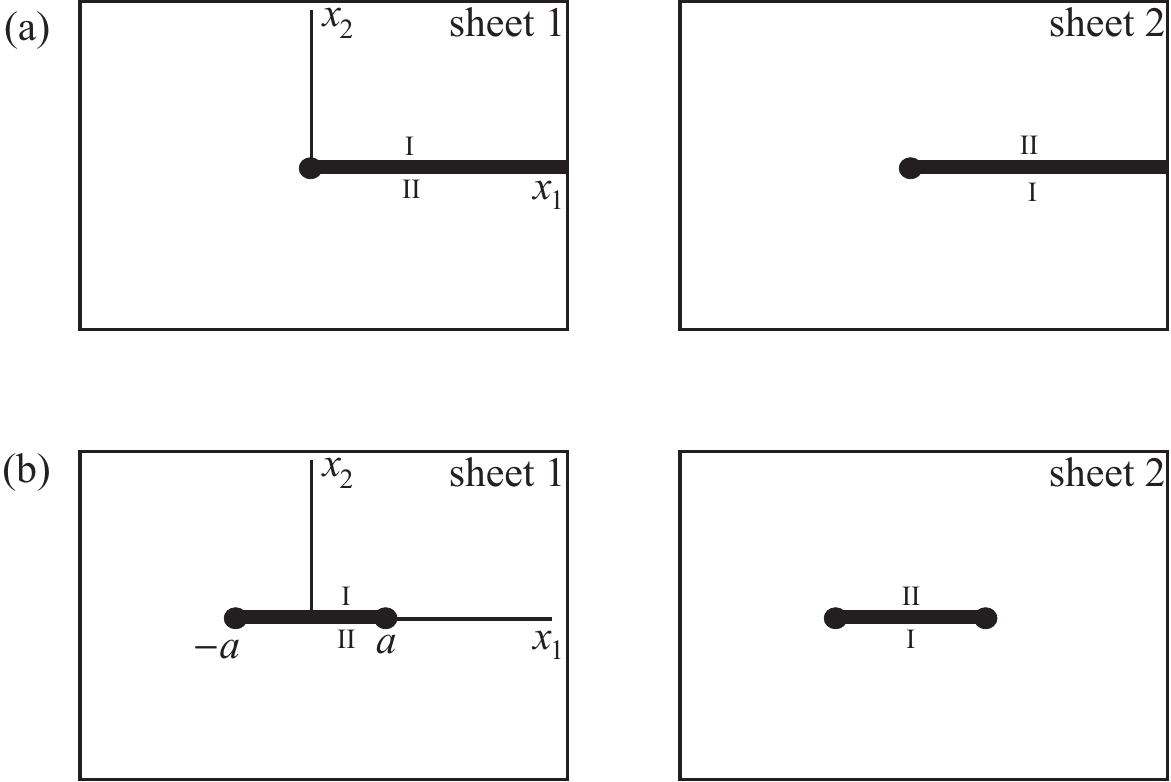}
\caption{Sommerfeld surfaces for the 2D problems of the Dirichlet half-line (a),
and the Dirichlet segment (b). The cuts are shown by thick lines. The sides of the cuts 
glued to each other bear the same Roman number. The associated branch points are denoted by thick black dots.}
\label{fig01}
\end{figure}

Sommerfeld surfaces emerge naturally when the reflection method is applied 
to a 2D diffraction problem with straight ideal boundaries. For example, the surfaces 
shown in Figure~\ref{fig01} help one to solve the classical Sommerfeld problem
of diffraction by a Dirichlet half-line 
\cite{Sommerfeld1896} and the problem of diffraction 
by a Dirichlet segment \cite{Shanin2002a}.   
A connection between the diffraction problems and the Sommerfeld surfaces is 
discussed in more details in Appendix~\ref{app:generalclass}, where the class of scatterers leading to finite-sheeted Sommerfeld surfaces is described. 

There exists a natural projection $\psi$ of a Sommerfeld surface $S$ 
to $\mathbb{R}^2$. For any small neighbourhood 
$U \subset \mathbb{R}^2$ not including any of the branch points 
$P_j$ the pre-image $\psi^{-1} (U)$ is a set of $M$ samples of $U$, as illustrated in Figure \ref{fig:illustrationsommefeldsurface} . 

\begin{figure}[h]
	\centering
	\includegraphics[width=0.4\textwidth]{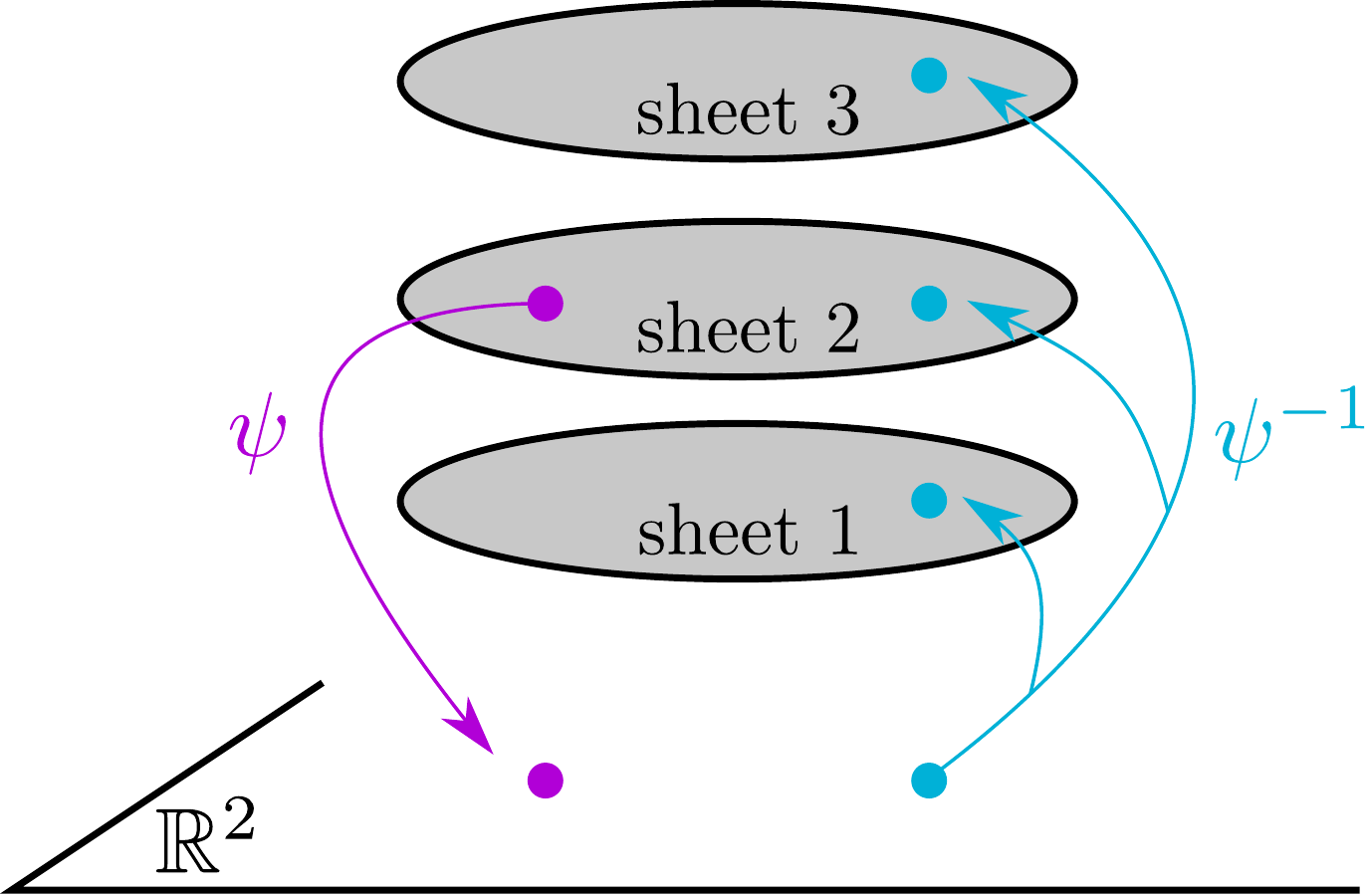}
	\caption{Diagrammatic illustration of the natural projection $\psi$ of a Sommerfeld surface $S$ with $M= 3$.}
	\label{fig:illustrationsommefeldsurface}
\end{figure}

Let $u$ be a continuous single-valued function on some Sommerfeld surface~$S$
(thus, $u$ is generally an $M$-valued function on $\mathbb{R}^2$).  
For any neigbourhood $U\subset\mathbb{R}^2\backslash\{P_1,\dots,P_N\}$, let $u$ obey the (real) Helmholtz equation 
\begin{equation}
(\ptl_{x_1}^2 + \ptl_{x_2}^2 + \mathscr{k}^2) u (x_1 , x_2) = 0,
\label{eq:Helmholtz}
\end{equation}
on each sample of $\psi^{-1} (U)$. The wavenumber parameter $\mathscr{k}$ is chosen to have a positive real part and a vanishingly 
small positive imaginary part mimicking damping of waves. 

Let $u$ also obey the Meixner condition
at the branch points $P_j$. The Meixner
condition characterises the local finiteness of the energy-type 
integral 
\[
\iint(|\nabla u|^2 + |u|^2) \, dx_1 dx_2,
\] 
near a branch point. 
In particular, it guarantees the absence of sources at the branch points of~$S$. 

A function $u$ obeying the equation (\ref{eq:Helmholtz}) in the sense explained above 
and the Meixner condition  will be referred to as a function obeying 
the Helmholtz equation on~$S$.

Let us now formulate a diffraction problem on a Sommerfeld surface $S$: find a function $u$ obeying the Helmholtz equation on $S$
that can be represented as a sum
\[
u = u^{\rm in} + u^{\rm sc},
\]
where $u^{\rm in}$ is a known incident 
field, which is equal to a plane wave or to zero, depending on the sheet: 
\begin{equation}
u^{\rm in} = \left\{ \begin{array}{l}
\exp \{ - i \mathscr{k}(x_1 \cos \vph^{\rm in} + x_2 \sin \vph^{\rm in})   \}, \\
0 .
\end{array} \right.
\label{eq:Uin}
\end{equation}
Here $\vph^{\rm in}$ is the angle of incidence. 
The incident wave is only non-zero on a single  sheet of $S$, and is zero 
on the other sheets. Neither $u^{\rm in}$ nor $u^{\rm sc}$ are continuous, 
but their sum is. 
The scattered field $u^{\rm sc}$ should also obey the limiting absorption principle, i.e.
be exponentially decaying as $\sqrt{x_1^2 + x_2^2} \to \infty$. 

In the rest of the paper, we assume that the existence and uniqueness theorem is proven for the chosen 
$S$ and~$\mathscr{k}$ and that the field $u$ is fully known on $S$.


\section{Analytical continuation of the field and its branching}
\label{sec:sec3analcont}

Our aim is to build an analytical continuation $u_{\rm c} (x_1 , x_2)$
of the solution $u(x_1, x_2)$ of a certain diffraction problem on a Sommerfeld 
surface~$S$. The continuation has the following sense. 

Let $x_1$ and $x_2$ be complex variables, i.e.\ $(x_1 , x_2) \in \mathbb{C}^2$.
Naturally, $\mathbb{C}^2$ is a space of real dimension~4.
Let $T \subset \mathbb{C}^2$ be a singularity set (built below), 
which is a union of several manifolds of real dimension~2.  
Let $T$ be such that the intersection of $T$ with the real plane 
$\mathbb{R}^2 \subset \mathbb{C}^2$ is the set of branch points~$P_1,\dots,P_N$.

The continuation $u_{\rm c}$ should obey three conditions.

\begin{itemize}

\item 
The continuation $u_{\rm c} (x_1 , x_2)$ 
should be a multivalued analytical function 
on $\mathbb{C}^2 \setminus T$. Each branch of $u_{\rm c}$ in any small 
domain $U\subset \mathbb{C}^2$ not intersecting with $T$ should obey the Cauchy--Riemann
conditions
\begin{equation}
\bar \ptl_{x_\ell} u_{\rm c} = 0, 
\quad \text{where} \quad
\bar \ptl_{x_\ell} \equiv \frac{1}{2} \left( 
  \frac{\ptl}{\ptl {\rm Re}[x_\ell]} + 
i \frac{\ptl}{\ptl {\rm Im}[x_\ell]}
\right),
\quad 
\ell \in \{ 1,2 \} .
\label{eq:Cauchy}
\end{equation}

\item
The continuation $u_{\rm c} (x_1 , x_2)$ should obey the complex Helmholtz equation 
in $\mathbb{C}^2 \setminus T$:  
\begin{equation}
(\ptl_{x_1}^2 + \ptl_{x_2}^2 + \mathscr{k}^2) u_{\rm c}  = 0,
\quad
\mbox{where}
\quad 
\ptl_{x_\ell} \equiv \frac{1}{2} \left( 
  \frac{\ptl}{\ptl {\rm Re}[x_\ell]} - 
i \frac{\ptl}{\ptl {\rm Im}[x_\ell]}
\right),
\quad 
\ell \in \{ 1,2 \}.
\label{eq:ComplexHelmholtz}
\end{equation}
Note that the notation $\ptl_{x_\ell}$ is used both in the real (\ref{eq:Helmholtz})
and in the complex (\ref{eq:ComplexHelmholtz}) context. However 
one can see that if the Cauchy--Riemann conditions are valid, the complex derivative 
gives just the same result as the real one. 

\item 
When considering the restriction of $u_{\rm c}(x_1 , x_2)$ onto $\mathbb{R}^2$, 
 over each point of  $\mathbb{R}^2 \setminus \{P_1, \dots P_N \}$, there should exist 
$M$ branches of $u_{\rm c}$ equal to the values of $u$ on $S$.

\end{itemize}

Let us describe, without a proof, the structure of the singularity set $T$ and the branching structure of $u_{\rm c}$. Later on, we shall build $u_{\rm c}$ explicitly, and one will be able to check the correctness of these statements. 

The branch points $P_j$ are singularities of the field $u$ on the real plane. 
According to the general theory of partial differential equations, 
the singularities propagate along the characteristics of the PDE (the Helmholtz equation here). 
Thus, it is natural to expect that $T$ is the union of the characteristics 
passing through the points $P_j=(X_1^{(j)},X_2^{(j)})\in\mathbb{R}^2$. Since the Helmholtz equation is elliptic, these characteristics are complex. They are given for $j\in{1,\dots,N}$ by 
\begin{eqnarray}
  L_1^{(j)} & = & \{ (x_1, x_2) \in \mathbb{C}^2, x_1 +i x_2 =
  X_1^{(j)} + iX_2^{(j)} \},  \label{eq:branch2lineL1j}\\
  L_2^{(j)} & = & \{ (x_1, x_2) \in \mathbb{C}^2, x_1 -i x_2 =
  X_1^{(j)} - iX_2^{(j)} \} .  \label{eq:branch2lineL2j}
\end{eqnarray}  
One can see that $L_{1,2}^{(j)}$ are complex lines having real dimension~2. 
We will hence refer to them as {\em 2-lines}. Their intersection with $\mathbb{R}^2$ is the 
set of the points~$P_j$, i.e.\ $L_1^{(j)}\cap L_2^{(j)}=P_j$. 

We demonstrate below that the 2-lines $L_{1,2}^{(j)}$ are, generally, branch 2-lines of 
$u_{\rm c}$. The branching of functions of several complex variables is not a well-known matter, thus, we should explain what it means. Consider for example a small neighbourhood 
$U \subset \mathbb{C}^2$ of a point on $L_1^{(1)}$,
which is not a crossing point of two such lines. 
Note that the complex variable
\[
z_1= x_1 +i x_2 - (X_1^{(1)} + iX_2^{(1)})
\]
is a coordinate transversal to $L_1^{(1)}$. 
The 2-line $L_1^{(1)}$ corresponds to $z_1 = 0$. 
The complex variable 
\[
z_2 = x_1 - i x_2
\] 
is then a coordinate tangential to $L_1^{(1)}$.  

Let $A$ be some point in $U$. Consider a path (oriented contour) 
$\sigma$ in $U$ starting and ending at $A$, and having no intersections with~$L_1^{(1)}$. Such a contour, called a \textit{bypass} of $L_1^{(1)}$, can be projected onto the 
variable~$z_1$. Denote this projection by $\sigma'$.  
One can continue $u_{\rm c}$ along $\sigma$ and obtain the branch 
$u_{\rm c}(A ; \sigma)$. Branches can be indexed by an integer $p$, which is the winding number of $\sigma'$ about zero.
If for some $\sigma'$ having winding number~$p$ 
\[
u_{\rm c}(A; \sigma) = u_{\rm c}(A)
\]    
for any such continuation (here we consider the smallest possible $p$ having this 
property), then the branch line $L_1^{(1)}$ has order of branching equal to $p$. 
If there is no such $p$, the branching is said to be logarithmic. 

Thus, generally speaking, the branching of a
function of several complex variables is similar to that of a single variable, and it is convenient to study this branching using a transversal complex coordinate. To provide 
the existence of such a transversal variable, the branch set should be a set (a complex manifold) of complex codimension~1.

For $j,k\in \{1,\dots N\}$, the 2-lines $L_1^{(j)}$
and $L_2^{(k)}$ intersect at a single point. For example, if $j = k$, this point is $P_j$, while for $j\neq k$, this intersection point does not belong to $\mathbb{R}^2$. The branching of 
$u_{\rm c}$ near each crossing point has a property that is new comparatively 
to the 1D complex case: the bypasses about $L_1^{(j)}$ and $L_2^{(k)}$ commute.

Let us prove this in the case $j=k$ by considering a small neighbourhood $U \subset \mathbb{C}^2$ of the point $P_j=L_1^{(j)}\cap L_2^{(j)}$. The case $j\neq k$ is similar.
Introduce the local coordinates 
\[
z_1 = x_1 +i x_2 - (X_1^{(j)} + iX_2^{(j)}),
\qquad 
z_2 = x_1 -i x_2 - (X_1^{(j)} - iX_2^{(j)}),
\]
which are transversal variables to 
$L_1^{(j)}$ and $L_2^{(j)}$ respectively. 
Take a point $A \in U \setminus (L_1^{(j)} \cup L_2^{(j)})$ and a path 
$\sigma$ in 
$U \setminus (L_1^{(j)} \cup L_2^{(j)})$ starting and ending at $A$.
Consider the projections $\sigma_1$ and $\sigma_2$ of $\sigma$ onto the 
complex planes of $z_1$ and $z_2$ respectively.

Assume that the path $\sigma$ is parametrised 
by a real parameter $\tau \in [0,1]$, i.e.\ 
\[
\sigma \, : \quad (x_1 (\tau) , x_2 (\tau)), 
\]
or, in the new coordinates,  
\[
\sigma \, : \quad (z_1 (\tau) , z_2 (\tau)).  
\]
The path $\sigma$ can be deformed homotopically into a path $\sigma^*$ defined by 
\[
\sigma^* \, : \quad (\epsilon z_1 (\tau)/ |z_1 (\tau)| ,
                     \epsilon z_2 (\tau)/ |z_2 (\tau)|),  
\]
for some small parameter $\epsilon$. The projection of $\sigma^*$ onto $z_1$ (resp.~$z_2$) is a small circle $\sigma^*_1$ (resp.~$\sigma^*_2$) of radius $\epsilon$ turning (possibly many times) around the origin. Therefore, $\sigma^*$ lies on a torus (product of two circles), for which $\sigma_1^*$ and $\sigma_2^*$ are strictly longitudinal and latitudinal paths respectively. Thus these loops commute (this comes from the fact that the fundamental group of the torus is Abelian). Therefore, the path 
$\sigma $ 
can be homotopically deformed into the concatenations 
$
\sigma = \sigma_1 \sigma_2 = \sigma_2 \sigma_1, 
$ 
where the path $\sigma_1$ occurs for fixed $z_2$, and $\sigma_2$ occurs for 
fixed~$z_1$. 


\section{Integral representation of the analytical continuation }
\label{sec:intrepofanal}

Here we are presenting the technique for analytical continuation of the wave field utilising Green's third identity as in e.g.\ \cite{Garabedian1960,Sternin1994}. 

Let $U \subset (S \setminus \cup_j P_j)$ be a small neighbourhood of a point   
$A_0 \in S$ such that $\psi(A_0)\!\!=\!\!(x_1,x_2)\!\!\in\!\!\mathbb{R}^2$, where $\psi$ is the natural projection of $S$ to $\mathbb{R}^2$. In what follows, in an abuse of notation, we may sometimes identify $A_0$ and $\psi(A_0)$ when the context permits to do so. 
Let the contour $\gamma$ be the boundary of $U$ 
oriented anticlockwise with unit external normal vector~$\boldsymbol{n}$.
Write Green's third identity for $A_0 \in U$:
\begin{eqnarray}
  u (A_0) & = & \int_{\gamma} \left[ \frac{\ptl G}{\ptl n'}
  ({\bf r}, {\bf r}') u ({\bf r}') - 
    \frac{\ptl u}{\ptl n'} ({\bf r}') G ({\bf r}, {\bf r}') \right]
  dl', \label{eq:initialgreensidentity}
\end{eqnarray}
where ${\bf r}'=(x_1',x_2')$ is a position vector along $\gamma$,
${\bf r}=(x_1,x_2)$ points to $A_0$,
 $\ptl /
\ptl n'$ corresponds to the normal derivative associated to the unit external normal
vector, and 
\begin{equation}
  G ({\bf r}, {\bf r}') = - \frac{i}{4} 
  H_0^{(1)} (\mathscr{k} \, r({\bf r}, {\bf r}')) \quad \text{ with } \quad  
  r ({\bf r},
  {\bf r}') = \sqrt{(x_1 - x_1')^2 + (x_2 - x_2')^2},
  \label{eq:Green}
\end{equation}
where $H_0^{(1)}$ is the zeroth-order Hankel function of the first kind. Note
that the point $A_0$ is the only singularity of
the integrand of (\ref{eq:initialgreensidentity}) in the real $(x_1', x_2')$
plane.

The orientation of $\gamma$ plays no role in (\ref{eq:initialgreensidentity}), 
however we can use graphically the orientation of contours to set the direction of the 
normal vector. Namely, let the normal vector point {\em to the right\/} from an oriented contour.  

The formula  (\ref{eq:initialgreensidentity}) can be used to continue $u(x_1 , x_2)$ to 
$u_{\rm c} (x_1 , x_2)$ in a small domain of $\mathbb{C}^2$. 
Namely let $A\equiv{\bf r}=(x_1, x_2)\in\mathbb{C}^2$ be complex, while $(x_1' , x_2')$ remains 
real and belongs
to~$\gamma$. If $(x_1 , x_2)$ is close to $A_0$, the Green's function 
$G({\bf r}, {\bf r'})$ remains regular for each ${\bf r'} \in \gamma$.
Moreover, being considered as a function of ${\bf r}$, the Green's 
function $G({\bf r} , {\bf r}')$ obeys the Cauchy--Riemann conditions
(\ref{eq:Cauchy}) and the complex Helmholtz equation (\ref{eq:ComplexHelmholtz})
provided ${\bf r}$ is a regular point for certain fixed~${\bf r}'$.
Thus, for a small complex neighbourhood of $A_0$ the formula 
(\ref{eq:initialgreensidentity}) provides 
a function obeying all restrictions (listed in Section \ref{sec:sec3analcont}) imposed on 
$u_{\rm c} (A)$.
 
We should note that the continuation $u_{\rm c}$ of $u$ in a small complex neighbourhood of $A_0$
is unique 
and is provided by letting ${\bf r}$ become a complex vector in (\ref{eq:initialgreensidentity}). 
The proof is given in Appendix~\ref{app:Greens} and its structure is as follows. We start by deriving a complexified Green's formula obeyed by $u_{\rm c}$
(or by any analytical solution of (\ref{eq:ComplexHelmholtz})). 
Then, by Stokes' theorem, the integration 
contour for this formula for $A$ belonging to some small complex neighbourhood of $A_0$
can be taken to coincide with~$\gamma$. In this 
case, the complexified Green's formula for $u_{\rm c}$ coincides with     
(\ref{eq:initialgreensidentity}).

The procedure of analytical continuation of $u$ using (\ref{eq:initialgreensidentity}) 
fails when $G({\bf r}, {\bf r}')$ becomes singular for some ${\bf r}'$. 
Let us develop a simple graphical tool to explore the singularities of 
$G({\bf r}, {\bf r}')$ for complex~${\bf r}\equiv A$. 
The function (\ref{eq:Green}) is singular when 
\begin{equation}
(x_1 - x_1')^2 + (x_2 - x_2')^2 = 0,
\label{eq:R0}
\end{equation}
i.e.\ when $A\equiv (x_1,x_2)$ and $A'\equiv (x_1',x_2')$ both belong to some characteristic of  (\ref{eq:ComplexHelmholtz}).
Let us fix the complex point $A$ and find the real points $A'\equiv (x_1' , x_2' )$
such that (\ref{eq:R0}) is valid. 
Obviously, $A'$ can have two values 
\begin{eqnarray}
  A_1 & = & ({\rm Re} [x_1] - {\rm Im} [x_2], {\rm Im} [x_1] + {\rm Re}
  [x_2]),  \label{eq:firstrealpoint}\\
  A_2 & = & ({\rm Re} [x_1] + {\rm Im} [x_2], {\rm Re} [x_2] - {\rm Im}
  [x_1]).  \label{eq:secondrealpoint}
\end{eqnarray}
These two points coincide when $A$ is a real point. 
We will call $A_1$ and $A_2$ the {\em first and the second real points associated 
with~$A$} and will sometimes use the notation $A_1(A)$ and $A_2(A)$ to emphasise the link between $A_{1,2}$ and $A$.
Both points $A_1$ and $A_2$ belong to~$\mathbb{R}^2$. Indeed, beside $A_1$ and $A_2$, one can 
consider their preimages $\psi^{-1} (A_1)$ and $\psi^{-1} (A_2)$ on~$S$.

The Green's function $G({\bf r} , {\bf r}')$ is singular at some point of the 
integral contour $\gamma$ in (\ref{eq:initialgreensidentity})
if $A_1(A) \in \gamma$ or $A_2(A) \in \gamma$, where ${\bf r}$ points to~$A$. 

Consider the analytical continuation along a simple path $\sigma$ as a 
continuous process. Let $\sigma$ be parametrised by a real parameter 
$\tau \in [0,1]$, i.e.\ let each point on the contour 
$\sigma$ be denoted by $A(\tau)$. 
Let $A(0) = A_0 \in S$ 
be the starting real point, and let $A(1)$ be the ending complex point.
For each point $A(\tau)$ find the associated real 
points $A_{1,2} (\tau) \equiv  A_{1,2} (A(\tau))$. The position of these points depends continuously on 
$\tau$. 

We are now well-equipped to formulate the first theorem of analytical continuation. 

\begin{theorem}
\label{th:Cont}
Let $A_0 \in S \setminus \cup_jP_j$ and let $A\in \mathbb{C}^2$ be within a small neighbourhood of $A_0$. Let $\sigma$ be a simple path from $A_0$ to $A$ parametrised by $\tau\in[0,1]$ as above. Let $\Gamma (\tau) \subset S$ 
be a continuous set of closed smooth oriented contours (i.e.\ a homotopical 
deformation of a contour) such that 
\begin{itemize}
\item
$ \Gamma(0) = \gamma $ ; 

\item
for each $\tau \in [0, 1]$
$A_1(\tau) \notin \psi (\Gamma(\tau))$, 
$A_2(\tau) \notin \psi (\Gamma(\tau))$.

\end{itemize}
Then the formula  
\begin{equation}
  u_{\rm c} (A(\tau) + \delta {\bf r})  =  \int_{\Gamma(\tau)} \left[ \frac{\ptl G}{\ptl n'}
  ({\bf r}, {\bf r}') u ({\bf r}') - 
    \frac{\ptl u}{\ptl n'} ({\bf r}') G ({\bf r}, {\bf r}') \right]
  dl', 
  \label{eq:ContFormula}
\end{equation}
defines an analytical continuation $u_{\rm c}$ of $u$ 
 in a narrow neighbourhood of~$\sigma$. $\delta {\bf r}$ is an arbitrary 
 small-enough complex radius vector, while
the radius-vector ${\bf r}$ points to $A(\tau) + \delta {\bf r}$. 
\end{theorem}
  
\begin{proof} 
We present a sketch of the proof on the ``physical level of rigour''. 
Let the contour $\Gamma(\tau)$ be changing 
incrementally, i.e.\ consider the contours $\Gamma(\tau_n)$, where 
$0 = \tau_0 < \tau_1 < \tau_2 \dots < \tau_K = 1$
is a dense grid on the segment $[0,1]$. 
Each fixed contour $\Gamma (\tau_n)$ provides an analytic function 
$u_{\rm c}$ in a small neighbourhood of the point $A(\tau_n)$. 
The grid is dense enough to ensure that such neighbourhoods are overlapping.
Moreover, for any point belonging to an intersection of neighbourhoods 
of $A(\tau_n)$ and $A(\tau_{n+1})$ one can deform the contour 
$\Gamma(\tau_n)$ into $\Gamma(\tau_{n+1})$ homotopically without changing the
value of the integral, and hence without changing the value of $u_{\rm c}$. 
\end{proof}

Note that, formally, the expression (\ref{eq:ContFormula})  defines the
field $u_{\rm c}$ ambiguously. The values of $u$ and $\ptl_{n'} u$ on $S$
are found in a unique way, but the values of $G$ and $\ptl_{n'} G$ should be 
clarified. Namely, according to (\ref{eq:Green}), the value depends on the 
branch of the square root and of the Hankel function (having a logarithmic 
branch point at zero). 

For $\tau = 0$,
let $G$ be defined on $\gamma = \Gamma (0)$ in the ``usual'' way: 
the square root is real positive, and the values of $H_0^{(1)}(\cdot)$ 
are belonging to the main branch of this function. Then, as $\tau$
changes continuously from $\tau = 0$ to $1$,   
define the values of $G$ and $\ptl_{n'} G$ by continuity.
Since the (moving) contour $\Gamma(\tau)$ does not hit the 
(moving) singular points $A_1 (\tau)$, $A_2 (\tau)$, the branch of $G$
is defined consistently. 

 
 The last theorem in this section extends the local result of Theorem \ref{th:Cont} to a global result.

\begin{theorem}
\label{th:Sing}
Let $B$ be a point of $\mathbb{C}^2 \setminus (T \cup \mathbb{R}^2)$, where 
$T$ is the union of all the 2-lines $L_{1,2}^{(j)}$. Let $A_0$ be a point belonging 
to $S \setminus (\cup_j P_j)$ and let $\sigma$ be a smooth  path 
connecting $A_0$ with $B$, such that 
\[
(\sigma \setminus A_0)  \cap (T \cup \mathbb{R}^2) = \emptyset.
\] 
Then there exists a family of contours  $\Gamma (\tau)$ associated with 
$\sigma$ and obeying the conditions
of Theorem~\ref{th:Cont}.  
\end{theorem}

We omit the proof of this theorem. 
It is almost obvious, but not easy to be formalised. 
One should consider the process of changing $\tau$ from 0
to~1, and ``pushing'' the already built contour $\Gamma(\tau)$, 
which is considered to be movable, 
by  the moving points $A_1 (\tau)$
and $A_2 (\tau)$ (or, to be more precise, by small disks
centred at $A_1 (\tau)$ and $A_2 (\tau)$).    

Some issues may occur with such deformation. For example, the contour may 
become pinched between $A_{1,2}(\tau)$ and $P_j$ for some $j$, or 
between $A_1(\tau)$ and $A_2(\tau)$.    
The condition 
\[
(\sigma \setminus A_0) \cap T = \emptyset 
\] 
guarantees that the points $A_1(\tau)$, $A_2 (\tau)$ do not 
pass through $P_j$, and thus the contour $\Gamma(\tau)$ cannot be pinched  
between $A_{1,2}$ and $P_j$. 
The condition 
\[
(\sigma \setminus A_0) \cap \mathbb{R}^2 = \emptyset 
\] 
guarantees that the contour $\Gamma(\tau)$ 
cannot be pinched between $A_1 (\tau)$ and $A_2 (\tau)$. 

Theorems~\ref{th:Cont} and \ref{th:Sing} demonstrate that $u_{\rm c}$ can be expressed  
almost everywhere as an integral (\ref{eq:ContFormula}) containing 
$u(x_1, x_2)$ defined on the real surface $S$, and some known kernel $G$. 
The continuation $u_{\rm c}$ has a more complicated structure 
than $u$ on $S$. This is explained by the fact that the set of 
closed contours on $S$ has a more complicated structure than $S$ itself. This will be investigated in the next section.


\section{Analysis of the analytical continuation branching}
\label{sec:analbranchinganal}

We will now study the branching of $u_{\rm c}$ as follows. 
Let us consider a starting point $A_0 \in (S \setminus \cup_j P_j$), 
an ending point $B$ located near a set on which we study the branching (i.e.\ near 
$\mathbb{R}^2$ or $L_{1,2}^{(j)}$), and a simple path $\sigma$ connecting $A_0$
to $B$ satisfying 
\[
(\sigma \setminus A_0)  \cap (T \cup \mathbb{R}^2) = \emptyset,
\] 
i.e.\ $\sigma$ is such that the condition of Theorem~\ref{th:Sing} is valid. 
According to the Theorems~\ref{th:Cont} and \ref{th:Sing},
one can continue $u_{\rm c}$ from $A_0$ to $B$ along~$\sigma$.
Let the result of this continuation be denoted by $u_{\rm c}(B)$.
Let $\sigma'$ be a local path starting and ending at $B$
and encircling a corresponding fragment of the branch set. For example, one can introduce 
a local transversal variable near the branch set and build a contour $\sigma'$
in the plane of this variable encircling zero for a single time in the positive direction. 
One can consider a continuation of $u_{\rm c} (B)$ along $\sigma'$ and obtain 
a branch of $u_{\rm c}$ that is denoted by $u_{\rm c} (B ; \sigma')$.  
This is a continuation of $u_{\rm c}$ from $A_0$
along the concatenation $\sigma \sigma'$ 
of the contours $\sigma$ and~$\sigma'$. 
Let the parameter $\tau$ parametrise the contour $\sigma \sigma'$; $\tau = 0$
correspond to $A_0$; $\tau = 1/2$ correspond to the end of $\sigma$;  
$\tau = 1$ correspond to the end of $\sigma  \sigma'$.
Consistently with this parametrisation, the contour of integration used to obtain $u_{\rm c} (B)$ and $u_{\rm c} (B ; \sigma')$ will be denoted by $\Gamma(1/2)$ and $\Gamma(1)$ respectively.

If $ u_{\rm c} (B ; \sigma') \equiv u_{\rm c} (B)$,
then the contour $\sigma'$ yields no branching. 
If $u_{\rm c} (B ; (\sigma')^p) \equiv u_{\rm c} (B)$, for some integer $p>1$,
then the branching has order $p$ (the smallest strictly positive integer with such property). Here $u_{\rm c} (B ; (\sigma')^p)$ is the result of 
continuation along the concatenation $\sigma  \sigma' \dots \sigma'$, 
where $\sigma'$ is taken $p$ times. 
The two following theorems establish the type of branching of $u_{\rm c}$.

\begin{theorem}
\label{th:Picard}
The points of the real plane $\mathbb{R}^2$ other than $P_j$ 
do not belong to the branch set of~$u_{\rm c}$. 
\end{theorem}

\begin{proof} 
 
Note that $\mathbb{R}^2$ is not an analytical set, so the concept of a transversal 
variable is not fully applicable to it (and of course this is 
the reason of non-branching at it). Anyway, the scheme sketched above can be applied 
to this case. Let $B$ be close to $\mathbb{R}^2$ but not close to any of the $P_j$.   
Consider a local path $\sigma'$ starting and ending at $B$ and encircling $\mathbb{R}^2$. 
Consider the motion of $A_1(A)$ and $A_2(A)$ as the point $A$ travels along~$\sigma'$. To be consistent with the parametrisation, we will refer to these points as $A_{1,2}(\tau)$ for $\tau\in[1/2,1]$. These points are close to eachother and close to $B$ since $B$ is close to 
$\mathbb{R}^2$.
Depending on the particular contour $\sigma'$,
it may happen that during this motion $A_2$ encircles $A_1$ a single time, 
several times or not at all.
Note that $A_1(1/2) = A_1 (1)=A_1(B)$ and $A_2(1/2) = A_2 (1)=A_2(B)$.  
 
One can show that 
if $A_2$ does not encircle $A_1$ then the contour $\Gamma(1)$ is homotopic 
to $\Gamma(1/2)$, and thus $u_{\rm c}(B ; \sigma') = u_{\rm c}(B)$. 

Let $A_2$ encircle $A_1$ a single time. 
The case of several times follows from this case 
in a clear way. To consider this case, we follow the principles of computation of ramified integrals that can be found for example in \cite{Pham2011, Sternin1994,Savin2017}. Namely, it is known  
that the integral changes locally during a local bypass. The 
fragments of the contour $\Gamma(1/2)$ that are located far from the points $A_1(1/2)$ and 
$A_2(1/2)$ are not affected by the bypass $\sigma'$. The fragments that are close to these 
points but do not pass between these points are not affected either. The only parts that are 
affected are the parts of $\Gamma(1/2)$ passing between the points $A_{1}(1/2)$ and $A_2 (1/2)$.   

In Figure~\ref{fig02} we demonstrate the deformation of a fragment of $\Gamma(1/2)$ 
(Figure~\ref{fig02}{\color{myred}a} into a corresponding fragment of $\Gamma (1)$ (Figure~\ref{fig02}{\color{myred}b}). 
One can see from Figure~\ref{fig02}{\color{myred}c} that the fragment of $\Gamma(1)$  
can be written as the sum of the initial fragment of $\Gamma(1/2)$ and 
an additional contour $\delta \Gamma$. 
The contour $\delta \Gamma$ is a ``double-eight'' contour having the
following property: it bypasses each of the points $A_1(1/2)$ and 
$A_2(1/2) $ 
zero times totally (one time in the positive and one time in the negative direction). 
We say that this double-eight contour is based on the points $A_1(1/2)$ and $A_2(1/2)$. Such contour is also known as a Pochhammer contour. 

\begin{figure}[ht]
\centering
\includegraphics[width=0.8\textwidth]{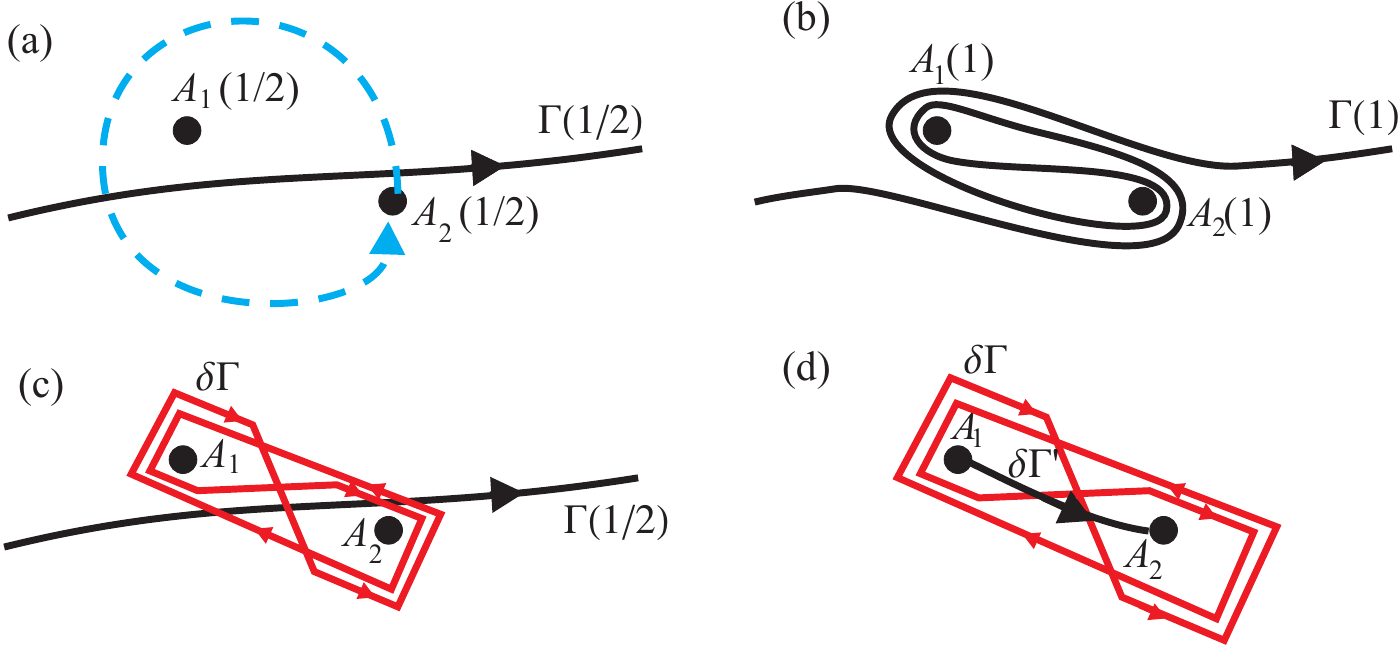}
\caption{Transformation of a fragment of the integration contour $\Gamma$ as
$A_2$ encircles $A_1$}
\label{fig02}
\end{figure}
 
Since there may be several parts of $\Gamma(1/2)$ passing between $A_1(1/2)$
and $A_2(1/2)$, the contour $\Gamma(1)$ can be written as a sum of $\Gamma(1/2)$ and several 
double-eight contours based on the points $A_1(1/2)$ and $A_2(1/2)$.
These contours, possibly, drawn on different sheets of $S$. 

Let us show that for a double-eight contour $\delta\Gamma$ based on $A_1(1/2)$ and $A_2(1/2)$, we have
\begin{equation}
   \int_{\delta \Gamma} \left[ \frac{\ptl G}{\ptl n'}
  ({\bf r}, {\bf r}') u ({\bf r}') - 
    \frac{\ptl u}{\ptl n'} ({\bf r}') G ({\bf r}, {\bf r}') \right]
  dl' = 0 .
  \label{eq:dGamma}
\end{equation}
This will yield that $u_{\rm c} (B; \sigma') = u_{\rm c} (B)$ 
and that there is no branching on 
$\mathbb{R}^2$. 
In order to do so, connect the points $A_1(1/2)$ and
$A_2(1/2)$ with an oriented contour $\delta \Gamma'$ (see Figure~\ref{fig02}{\color{myred}d}, 
$\delta \Gamma'$ is shown by a black line).
Squeeze the contour $\delta \Gamma$ in such a way that it goes $4$ times 
along $\delta \Gamma'$. 

 We now claim that it is not necessary to account for the integral in the close vicinity of $A_1(1/2)$ or $A_2(1/2)$. Indeed, using (\ref{eq:Green}), (\ref{eq:firstrealpoint}) and (\ref{eq:secondrealpoint}), one can show that the kernel $G$ has two logarithmic branch points at $A_1(1/2)$ and $A_2(1/2)$, and hence, in principle, the contour $\delta\Gamma$ needs to be considered on the Sommerfeld surface of $G$. Such surface has infinitely many sheets, and can be constructed by considering a straight cut between $A_1(1/2)$ and $A_2(1/2)$ and suitable ``gluing''. What is important is that locally, around $A_1(1/2)$ and $A_2(1/2)$, $G$ behaves like a complex logarithm. Hence, locally, the difference in $G$ between two adjacent sheets is constant. Therefore, when considering the part of $\delta\Gamma$ in the vicinity of $A_1(1/2)$ and $A_2(1/2)$, two loops going in opposite directions, the overall contribution of the integral tends to zero as these loops ``shrink'' to $A_{1,2}(1/2)$. This is due to the fact that, $u$ being single valued on $\delta \Gamma$, the logarithmic singularities of $G$ cancel out. Thus, a consideration of the 
	integral near $A_1(1/2)$ and $A_2(1/2)$ is not necessary, and it is possible to reduce $\delta \Gamma$ to four copies of $\delta \Gamma'$.
 
As explained in detail in Appendix \ref{app:bypasshankel}, according to the formulae (\ref{eq:firstrealpoint})-(\ref{eq:secondrealpoint}) linking the complex variables 
$(x_1 , x_2)$ with the real coordinates of the points $A_1 , A_2$,
and to formulae (\ref{eq:Green}), a bypass about $A_1$ in the positive direction (in the real plane)
leads to the following change of the argument of $H_0^{(1)} (z)$:
$z \to e^{i \pi} z$.
Similarly, a bypass about $A_2$ in the positive direction changes 
the argument of the Hankel function as  $z \to e^{-i \pi} z$.

As a result, one can rewrite the integral of (\ref{eq:dGamma})
as 
\begin{equation} 
 \int_{\delta \Gamma} \left[ \frac{\ptl G}{\ptl n'}
  ({\bf r}, {\bf r}') u ({\bf r}') - 
    \frac{\ptl u}{\ptl n'} ({\bf r}') G ({\bf r}, {\bf r}') \right]
  dl' =
\int_{\delta \Gamma'} \left[ \frac{\ptl G'}{\ptl n'}
  ({\bf r}, {\bf r}') u ({\bf r}') - 
    \frac{\ptl u}{\ptl n'} ({\bf r}') G' ({\bf r}, {\bf r}') \right]
  dl' ,
  \label{eq:dGammaPrime} 
\end{equation}
where 
\begin{equation}
G' ({\bf r} , {\bf r}') = - \frac{i}{4} \left( 
-H_0^{(1)}(e^{- i \pi}\mathscr{k}\, r({\bf r} , {\bf r}') )+ 
2 H_0^{(1)}(\mathscr{k}\, r({\bf r} , {\bf r}'))
- H_0^{(1)}(e^{ i \pi} \mathscr{k}\, r({\bf r} , {\bf r}') )
\right).
\label{eq:GP}
\end{equation}

The notation $H_0^{(1)} (e^{i \pi} z)$ means the value of $H_0^{(1)} (\cdot)$
obtained as the result of continuous rotation of the argument $z$ about the origin
for the angle $\pi$ in the positive direction.   

Apply the formula (see e.g.\ \cite{DSJones1964}, sec.\ 1.33, eqs (202)--(203)) well-known from the theory of Bessel functions :
\begin{equation}
-H_0^{(1)}(e^{- i \pi} z)+ 
2 H_0^{(1)}(z)
- H_0^{(1)}(e^{ i \pi} z )
= 0 ,
\label{eq:HankelId}
\end{equation} 
to conclude that
the expression on the right-hand side of (\ref{eq:dGammaPrime}) is equal to zero and conclude the proof. Here and below,  
the formula (\ref{eq:HankelId}) seems to play a fundamental role in the process 
of analytical continuation.    
 \end{proof}

	Note also that every branch of $u_c$ is continuous at any point $A_0 \in \mathbb{R}^2\setminus{\cup_jP_j}$, thus these points 
	are regular points of any branch of $u_{\rm c}$. This can be seen using the complexified Green's theorem of Appendix \ref{app:Greens}. Indeed, using the same idea of variable translation, it can be shown that for a point $A\in\mathbb{C}^2\setminus(T\cup\mathbb{R}^2)$ in a close complex neighbourhood of $A_0$ and a given branch $u_c(A)$, we can write $u_c(A)=\int_{\tilde{\gamma}} (u_c\nabla_CG-G\nabla_Cu_c)$ for some small contour $\tilde{\gamma}$ surrounding $A$ in $\mathbb{C}^2$. Since the differential form $u_c\nabla_CG-G\nabla_Cu_c$ is closed (see Appendix \ref{app:Greens}), Stokes' theorem allows us to deform $\tilde{\gamma}$ to a contour $\gamma$ in $\mathbb{R}^2$ surrounding $A_0$ (see Figure \ref{fig:appCfig}) without changing the value of the integral. $A$ is chosen close enough to $A_0$ so that we can let $A\to A_0$ along a simple small straight path without any singularities of the integrand hitting the contour $\gamma$, showing, in doing so, the continuity of $u_c$ at $A_0$.


 \begin{theorem}
 	\label{th:Picard2}
 	Let $P_j \in \mathbb{R}^2$ be a branch point of order $p$ (the definition 
 		of the order of the branch point is clarified below). Then 
 		both $L_1^{(j)}$ and $L_2^{(j)}$ are branch 2-lines of $u_{\rm c}$ of 
 		order~$p$. 
 \end{theorem}

\begin{proof}
First, let us define what it exactly means for $P_j \in \mathbb{R}^2$ to be a branch point 
of $S$ of order~$p$.
The situation is simple if $\psi^{-1} (P_j)$ is a single point. Then $p$ is the order of branching of $S$ at $\psi^{-1}(P_j)$. If $\psi^{-1} (P_j)$ is a set of several points, 
then $p$ is the least common multiple of orders of all points $\psi^{-1} (P_j)$. 
A bypass encircling $P_j$ $p$ times in the positive direction returns each 
point located closely to $P_j$ to itself. 

Consider a point $B\in\mathbb{C}^2\setminus T$ that is close to $L_1^{(j)}$. Let the loop $\sigma'$ bypass $L_1^{(j)}$ once in the positive direction 
(i.e.\ its projection onto the local transversal coordinate
\[
z_1 = x_1 + i x_2 - (X_1^{(j)} + i X_2^{(j)})
\]
bypasses zero once in the positive direction). 
Let $\sigma'$ be also such that it does not bypass any other 2-line of~$T$. As before, parametrise $\sigma'$ by $\tau\in[1/2,1]$ and denote by $\Gamma(\tau)$ the associated integration contours used for the analytical continuation along $\sigma'$ . Then $A_1(\tau)$ bypasses $P_j$ once in the positive direction, 
and $A_2(\tau)$ does not bypass any branch point. 

The fragments 
of $\Gamma(1/2)$ that are far from $A_1 (1/2)$ or not passing between 
$A_1 (1/2)$ and $P_j$ are not affected by~$\sigma'$.
For each fragment of $\Gamma (1/2)$ passing between 
$A_1 (1/2)$ and $P_j$,
the double-eight contour $\delta \Gamma$ is added to obtain a fragment of~$\Gamma (1)$. 
In this case, the double-eight contour is based on the points 
$A_1 (1/2)$ and~$P_j$. The graphical proof is similar to that in Figure~\ref{fig02}.  

Let us now show that the order of branching at $L_1^{(j)}$ is equal to~$p$.
Consider a fragment of $\Gamma(1/2)$ passing between $A_1$ and $P_j$.
A single bypass along $\sigma'$ changes (locally)
\[
\Gamma(1/2) \stackrel{\sigma'}{\longrightarrow} \Gamma(1/2) + \delta \Gamma,
\]
where $\delta \Gamma$ is illustrated in Figure \ref{fig04}. 
\begin{figure}[ht]
	\centering
	\includegraphics{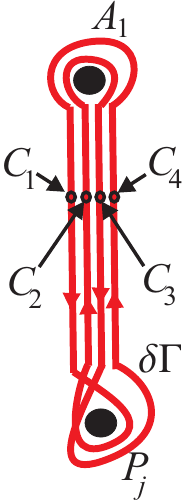}
	\caption{Illustration of $\delta \Gamma$ and the points of application of formula (\ref{eq:HankelId}) to $\delta \Gamma$}
	\label{fig04}
\end{figure}
Upon performing a second bypass $\sigma'$, we hence get 
\[
\Gamma(1/2) + \delta \Gamma
\stackrel{\sigma'}{\longrightarrow} \Gamma(1/2) + \delta \Gamma + \delta \Gamma^{(1)},
\]
where $\delta \Gamma^{(1)}$ is the double-eight contour 
based on the points $A_1$ and $P_j$, and
obtained from $\delta \Gamma$ by rotating $A_1$ about $P_j$ once in the positive direction. 

The projection  $\psi(\delta \Gamma^{(1)})$ coincides with 
$\psi (\delta \Gamma)$, however, the new contour passes along other sheets of 
$S$, and the branch of $H_0^{(1)}$ is different on it. 
Finally, after $p$ rotations one gets the integration contour 
$\Gamma(1/2) + \delta \Gamma + \delta \Gamma^{(1)} + \dots +\delta \Gamma^{(p-1)}$.
In order to show that the branching of $L_1^{(j)}$ is of order $p$, we need to show that that 
	\begin{equation} 
	\int_{\delta \Gamma + \delta \Gamma^{(1)} + \dots + \delta \Gamma^{(p-1)}} 
	\left[ \frac{\ptl G}{\ptl n'}
	({\bf r}, {\bf r}') u ({\bf r}') - 
	\frac{\ptl u}{\ptl n'} ({\bf r}') G ({\bf r}, {\bf r}') \right]
	dl' = 0 ,
	\label{eq:dGammaPrime1} 
	\end{equation}
implying that $u_{\rm c} (B ; (\sigma')^p) = u_{\rm c}(B)$, and that the branching has order~$p$. Algebraically, in terms of homology classes, this can be written as  
\begin{equation}
\delta \Gamma + \delta \Gamma^{(1)} + \dots + \delta \Gamma^{(p-1)} = 0 .
\label{eq:ContourEq}
\end{equation}

We will show this by decomposing $\delta \Gamma$ and all its subsequent ``copies'' into three main parts: (i) the two circles around $P_j$; (ii) the two circles around $A_1$ and (iii) the four straight lines between $A_1$ and $P_j$, and showing that the contribution of each of these three parts to the integral (\ref{eq:dGammaPrime1}) is indeed zero.
\begin{enumerate}[label=(\roman*)]

\item By Meixner conditions, $u$ and $\partial_n'u$ are integrable in the close vicinity of $P_j$, where $G({\bf r}, \cdot)$ is well behaved. Hence we can ``shrink'' the circles to $P_j$, and the integrals over the two circles become zero and do not contribute to (\ref{eq:dGammaPrime1}).

\item For the circles around $A_1$, the argument is slightly more subtle. Let $(s_1,\cdots,s_p)$ be the sheets of $S$ accessible by turning around $P_j$. Whatever sheet $u$ is on the initial outer circle, after $p$ rotation, it will have ``visited'' all sheets $s_1,\cdots,s_p$ once. The same is true for the inner circle. Hence, the overall contribution of the circles to (\ref{eq:dGammaPrime1}) can be written as
\begin{align}
&\int_{A_1} \left\{ \left[ \frac{\ptl G}{\ptl n'}
({\bf r}, {\bf r}';h_1) u ({\bf r}';s_1) - 
\frac{\ptl u}{\ptl n'} ({\bf r}';s_1) G ({\bf r}, {\bf r}';h_1) \right] \right.\nonumber\\
&\left. +\cdots+ \left[ \frac{\ptl G}{\ptl n'}
({\bf r}, {\bf r}';h_p) u ({\bf r}';s_p) - 
\frac{\ptl u}{\ptl n'} ({\bf r}';s_p) G ({\bf r}, {\bf r}';h_p) \right] \right\} \,dl'\nonumber\\
-&\int_{A_1} \left\{ \left[ \frac{\ptl G}{\ptl n'}
({\bf r}, {\bf r}';h_1') u ({\bf r}';s_1) - 
\frac{\ptl u}{\ptl n'} ({\bf r}';s_1) G ({\bf r}, {\bf r}';h_1') \right] \right.\nonumber\\
&\left. +\cdots+ \left[ \frac{\ptl G}{\ptl n'}
({\bf r}, {\bf r}';h_p') u ({\bf r}';s_p) - 
\frac{\ptl u}{\ptl n'} ({\bf r}';s_p) G ({\bf r}, {\bf r}';h_p') \right] \right\} \,dl',
\end{align} 
for some $(h_1,\cdots,h_p)$ and $(h_1',\cdots,h_p')$ corresponding to given sheets of $G({\bf r},\cdot)$, where the $;$ notation in the argument of a functions specifies which sheet it is on, and where $\int_{A_1}$ specifies integration along a small circle encircling $A_1$. Now we can group the terms for which $u$ lies on the same sheet. For each such pair, we can use the reasoning used below (\ref{eq:dGamma}) in the proof of Theorem \ref{th:Picard} to show that even if $G$ may be on a different sheet for each element of the pair, the logarithmic singularities cancel out. Hence one can safely ``shrink'' the circles to $A_1$ without leading to any contribution to (\ref{eq:dGammaPrime1}).

\item Finally, for the straight lines, one can check that 
due to the identity  (\ref{eq:HankelId}), 
the branch of $H_0^{(1)}$ (and hence of $G$) does not matter. Namely, let $C_1, C_2, C_3, C_4$ 
be some points of $\delta \Gamma$ projected onto a single point $C \in \mathbb{R}^2$.
These points are shown in Figure~\ref{fig04}, but for clarity they are shown close to each other, 
not above one another.

Consider the points $C_1$ and $C_4$. They belong to the same sheet of $S$, 
but the branches of $H_0^{(1)}$ are different at these points. Namely, if the argument of 
$H_0^{(1)}$ is equal to $z$ at $C_1$, then it is equal to $z e^{-i \pi}$ at $C_4$ (see again Appendix \ref{app:bypasshankel}).
Since the contours have different directions at $C_1$ and $C_4$, the contribution of the two outer lines to the integral over $\delta \Gamma$ is of the form
\begin{align}
\int_{\delta \Gamma'} \left[ \frac{\ptl G'}{\ptl n'}
({\bf r}, {\bf r}') u ({\bf r}';s) - 
\frac{\ptl u}{\ptl n'} ({\bf r}';s) G' ({\bf r}, {\bf r}') \right] \, dl',
\end{align}
for some $s\in(s_1,\cdots,s_p)$, where $G' ({\bf r}, {\bf r}')=-\tfrac{i}{4}(H_0^{(1)} (\mathscr{k}r({\bf r}, {\bf r}')) - H_0^{(1)} (e^{- i \pi}\mathscr{k}r({\bf r}, {\bf r}'))$ and where $\delta\Gamma'$ is a straight line between $A_1$ and $P_j$.

Perform a bypass $\sigma'$. In the course of this bypass, the points $C_1$ and $C_4$
encircle $A_1$ in the positive direction. Thus, the corresponding part of the integral 
along $\delta \Gamma^{(1)}$ contains 
$H_0^{(1)} (z e^{i \pi}) - H_0^{(1)} (z )$ and its derivative. 
However, due to (\ref{eq:HankelId}), 
\begin{equation}
 H_0^{(1)} (z e^{i \pi}) - H_0^{(1)} (z ) = 
 H_0^{(1)} (z) - H_0^{(1)} (z e^{- i \pi}),
\label{eq:HankelId1}
\end{equation}
and hence the value of $G'$ is unaffected by such bypass.
As before, after $p$ bypasses, $u$ would have ``visited'' all the sheets $(s_1,\cdots,s_p)$ and the overall contribution of the two outer lines to (\ref{eq:dGammaPrime1}) can be written
\begin{align}
\int_{\delta \Gamma'} \left[ \frac{\ptl G'}{\ptl n'}
({\bf r}, {\bf r}') u' ({\bf r}') - 
\frac{\ptl u'}{\ptl n'} ({\bf r}') G' ({\bf r}, {\bf r}') \right] \, dl',
\label{eq:contrib59}
\end{align} 
where $u'(\cdot)=u(\cdot,s_1)+\cdots+u(\cdot,s_p)$.

The same consideration can be conducted for the points $C_2$ and~$C_3$ on the inner lines to show an overall contribution equal to minus that of (\ref{eq:contrib59}). Hence the overall contribution of the straight lines to (\ref{eq:dGammaPrime1}) is also zero.

\end{enumerate}

 We have therefore proved that the equality (\ref{eq:dGammaPrime1}) is correct, and hence that $L_1^{(j)}$ is a branch 2-line of order $p$.
The case of $\sigma'$ bypassing $L_2^{(j)}$ can be considered in a similar way. Note that this time, if $\sigma'$ bypasses $L_2^{(j)}$ once in the positive direction (and no other branch 2-line), then the corresponding $A_2 (\tau)$ bypasses $P_j$ once in the negative direction, and $A_1 (\tau)$ does not bypass any branch point.
\end{proof} 

\begin{theorem}
\label{the:PrincipleAC}
Consider two points $A, B \in (\mathbb{C}^2 \setminus T)$, and let $\sigma_1$ and $\sigma_2$ be 
piecewise-smooth paths in $(\mathbb{C}^2 \setminus T)$
connecting $A$  with~$B$. Assume that it is possible 
to continue homotopically $\sigma_1$ to $\sigma_2$ in $(\mathbb{C}^2 \setminus T)$.
Let $u_{\rm c}(A)$ be some branch of $u_{\rm c}$ in some neighbourhood
of~$A$. Then the branches of $u_{\rm c}(B)$ obtained by continuation of $u_{\rm c}(A)$ along 
$\sigma_1$ and along $\sigma_2$ coincide:
\[
u_{\rm c}(B; \sigma_1) = u_{\rm c}(B ; \sigma_2). 
\]   
\end{theorem}

This theorem follows naturally from the principle of analytical continuation, 
which remains valid in 2D complex analysis. In the following theorem, we show that any closed contour in $\mathbb{C}^2\setminus T$ can be deformed to a concatenation of \textit{canonical building blocks}.

\begin{theorem}
\label{th:ContourRep}
Let $A$ be a point in $\mathbb{C}^2 \setminus T$, and let $\sigma$ be a closed 
path in $\mathbb{C}^2 \setminus T$ starting and ending at $A$. Then, within $\mathbb{C}^2 \setminus T$,  $\sigma$ can be homotopically transformed into a concatenation contour $\sigma'$ of the form 
\begin{equation}
\sigma' = \sigma_{\alpha_1}^{m_1} \sigma_{\alpha_2}^{m_2} \dots \sigma_{\alpha_H}^{m_H},
\label{eq:Contour1}
\end{equation}
for some positive integer $H$. For $h\in\{1,\dots,H\}$, the powers $m_h$ belong to $\mathbb{Z}$ and 
the elementary contours $\sigma_{\alpha_h}$
are fixed for each $\alpha_h$. 
The indices $\alpha_h$ are pairs of the type $(\ell,j)$, where $\ell \in \{ 1,2\}$ and  
$j \in \{1, \dots, N \}$.
These contours can themselves be represented as 
concatenations
\begin{equation}
\sigma_{\alpha_h} = \gamma_{\alpha_h} \sigma^*_{\alpha_h} 
\gamma^{-1}_{\alpha_h} ,
\label{eq:Contour2}
\end{equation} 
where each contour $\gamma_{\alpha_h}$ goes from $A$ to some point near $L_{\ell}^{(j)}$,
and $\sigma^*_{\alpha_h}$ is a local contour encircling $L_{\ell}^{(j)}$ once. 
On each contour $\sigma_{1,j}$ the value $x_1 - ix_2$ is constant. 
On each contour $\sigma_{2,j}$ the value $x_1 + ix_2$ is constant. 
\end{theorem}

Examples of such contours can be found in Figure~\ref{fig05}.   

\begin{figure}[ht]
\centering
\includegraphics{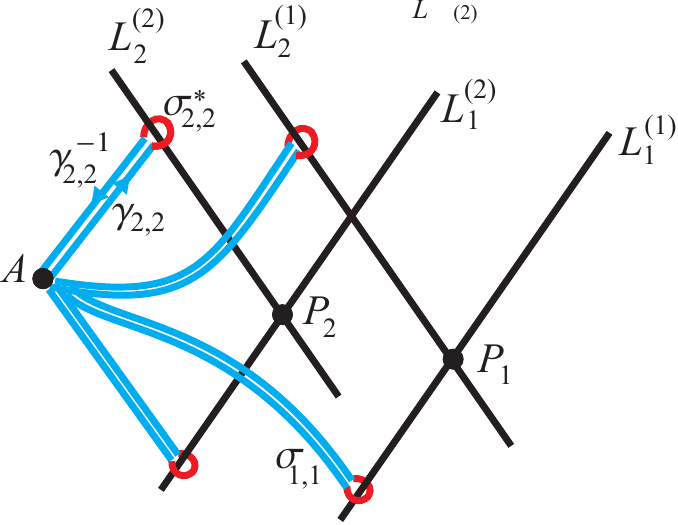}
\caption{Illustration of some contours $\sigma_\alpha$, $\gamma_\alpha$ and $\sigma^*_\alpha$}
\label{fig05}
\end{figure}

\begin{proof}
	Consider a small neighbourhood $U\subset\mathbb{C}^2\setminus T$ of some point $A$. By the principles of analytic continuation, all possible branches $u_c(A)$ are obtained by continuation along loops ending and starting at $A$ that cover all possible combinations of the homotopy classes of $\mathbb{C}^2\setminus T$. Each $\sigma_\alpha$ represents one of the homotopy classes, and hence by allowing $\sigma'$ to take the form (\ref{eq:Contour1}), all possible combinations of such classes are covered. Finally since any closed contour $\sigma$ can be written as a combination of homotopy classes, then we can in principle deform any closed contour $\sigma$ to one akin to $\sigma'$.
\end{proof}


Note that the elementary paths $\sigma_{\alpha}$ have the following property. 
Since either $x_1 - i x_2$ or $x_1 + i x_2$ is constant on such a path, 
either $A_2$ or $A_1$ remains constant as the path $\sigma_{\alpha}$ is passed. 

In the next theorem, we formulate the main general result of the paper, which is that there exists a finite basis of elementary functions such that any branch of the analytical continuation $u_c$ is a linear combination, with integer coefficients, of such functions.



\begin{theorem}
\label{th:FiniteBasis}
For a neighbourhood $U\subset\mathbb{C}^2\setminus T$ of a given point $A$, one can find a finite set of basis 
functions 
$g_1 (A) , \dots , g_Q (A)$, 
$A \in U$,
which are analytical 
solutions of the complex Helmholtz equation (\ref{eq:ComplexHelmholtz}) in $U$, 
and such that any branch $u_{\rm c}(A ; \sigma)$,  
can be written as a linear combination 
\begin{equation}
u_{\rm c}(A ; \sigma) = \sum_{q = 1}^Q
b_q (\sigma) g_q (A) ,
\label{eq:LinComb}
\end{equation}
where $b_q (\sigma) $ are integer coefficients that are 
constant with respect to $A$.  
The dimension of the basis, $Q$, can be defined from the topology of $S$.
\end{theorem}

\begin{proof}
 
According to Theorem~\ref{th:Sing}, Theorem~\ref{th:Picard}, Theorem~\ref{th:Picard2}  and the basic principles 
of analytical continuation, it is enough to prove the statement of the theorem for an arbitrary small neighbourhood $U$. 

Indeed, if the theorem is true for such a neighbourhood, the elements of the basis, then, can be expressed 
as linear combinations of $Q$ linearly independent branches of $u_{\rm c}$.
The coefficients of the combination are constant, therefore the elements of the basis 
have singularities only at $T$, and the type of branching is the same as that of $u_{\rm c}$. Any other neighbourhood $U'$ can be connected 
with $U$ with a path in $\mathbb{C}^2 \setminus T$, and the formula 
(\ref{eq:LinComb}) can be continued along this path.
 
Hence, below, we prove the theorem for a point $A$, for which the mutual 
location of the points $A_{1,2}$ (the real points associated to $A$) and $P_j$ (the branch points on $S$) is convenient in some sense. 
 
Choose the ``convenient'' neighbourhood 
$U$ as follows. Let us assume that the point $A\equiv(x_1, x_2)$ is far enough
from the branch points $P_j$, is close to $\mathbb{R}^2$ 
but does not belong to $\mathbb{R}^2$. This results in a configuration akin to that illustrated in Figure~\ref{fig03}.    
  
\begin{figure}[ht]
\centering
\includegraphics[width=0.4\textwidth]{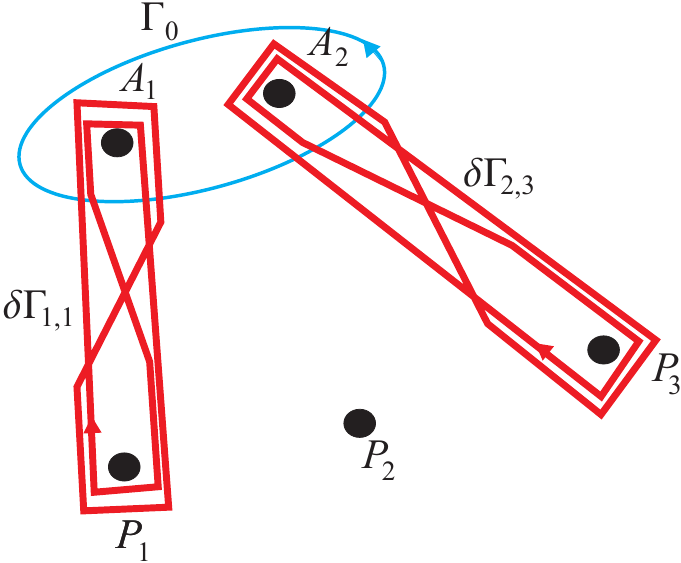}
\caption{Illustration of some integration contours used to define the basis functions}
\label{fig03}
\end{figure}

Define the basis functions $g_q$ as follows. Let $g_1$ be defined by 
the integral (\ref{eq:ContFormula}) with the
contour 
$\Gamma_0$ shown in Figure~\ref{fig03}. Such a function $g_1$
is equal to $u_{\rm c}$ obtained by analytical continuation from a small neighbourhood of a point belonging to~$S$ as per Theorem \ref{th:Cont}. 

All the other basis functions are constructed as follows.  
Let $\delta \Gamma_{\ell,j}$ be a double-eight contour based on the 
points $A_\ell$, $\ell \in \{ 1,2 \}$, and $P_j$, $j \in \{1, \dots, N \}$, as illustrated in Figure~\ref{fig03}. 
Consider all preimages $\psi^{-1}(\delta \Gamma_{\ell,j})$ on~$S$.
Denote them by  $\delta \Gamma_{\ell,j}^{(1)},\dots,\delta \Gamma_{\ell,j}^{(M)}$, where we remind the reader that $M$ is the finite number of sheets of $S$. Some of them 
are linearly independent.
The functions $g_2 , \dots , g_Q$ are the integrals of the 
form (\ref{eq:ContFormula}) taken with all linearly independent 
contours from the set~$\{\delta \Gamma_{\ell,j}^{(1)},\dots,\delta \Gamma_{\ell,j}^{(M)}\}$. 

To see this, let us continue the function $u_{\rm c}$ from $A$ along some closed
path~$\sigma$. Deform the path $\sigma$ into a path $\sigma'$ as in Theorem~\ref{th:ContourRep}. 
Each building block of $\sigma'$, $\sigma_{\ell,j}$, can be analysed by the procedure 
described in the proof of Theorem~\ref{th:Picard2}.     
A local contour $\sigma^*_{\ell,j}$ produces several 
local 
double-eight loops, and the path $\gamma^{-1}_{\ell,j}$
stretches the loops into those shown in Figure~\ref{fig03}. 
\end{proof}

As mentioned in the statement of Theorem \ref{th:FiniteBasis}, the exact value of the dimension $Q$ of the basis depends on the topology of $S$ and hence on the specific diffraction problem considered. In the next section we show (among other results) that in the case of the Dirichlet strip problem, we have $Q=4$.

\section{The strip problem}
\label{sec:strip6}

Everything so far has been done for a generic scattering problem described in introduction. We
will here deal with the specific problem of diffraction by a finite strip $(-
a < x_1 < a, x_2 = 0)$
. The canonical problem of
diffraction by a strip has attracted a lot of attention since the beginning of
the 20th century, and various innovative mathematical methods have been
designed and implemented to solve it: Schwarzschild's series
{\cite{Schwarzschild1901}}, Mathieu functions expansion {\cite{Morse1938}},
modified Wiener-Hopf technique {\cite{Noble1958,Abrahams1982}}, embedding and
reduction to ODEs {\cite{WILLIAMS1982,Shanin2001}}. It has important
applications, including in aero- and hydro-acoustics, see {\cite{Nigrophd}} for
example.

The problem can be formulated as follows: find the total field $u$ satisfying
the Helmholtz equation and Dirichlet boundary conditions $(u = 0)$ on the
strip, resulting from an incident plane wave. The scattered field (the
difference between the total and the incident field) should satisfy the
Sommerfeld radiation condition, and the total field should satisfy the Meixner
conditions at the edges of the strip.

Here we assume that this physical field $u$ is known and we consider the associated Sommerfeld surface $S$. It has two branch points $P_1 = (a, 0)$ and $P_2 = (- a, 0)$, so that $N = 2$. They are each of order 2, and the Sommerfeld surface $S$ has two sheets ($M= 2$). The surface $S$ is shown in Figure~\ref{fig01}{\color{myred}b}. Let us assume that sheet~1 is  the physical sheet, while sheet~2 is its mirror reflection. 

 We will now apply the general theory developed in the paper in order to unveil the analytical continuation $u_c$.



Let $A\equiv(x_1,x_2)\in\mathbb{C}^2\setminus(T\cup\mathbb{R}^2)$ be some point near $\mathbb{R}^2$.
Consider all possible continuations of $u_{\rm c}(A)$ along closed paths.
According to Theorem~\ref{th:ContourRep},
any such path can be represented as a concatenation of
elementary paths $\sigma_{\ell,j}$ for $\ell\in\{1,2\}$ and $j\in\{1,2\}$.    
The elementary paths 
can be chosen in a such a way that corresponding trajectories of the points 
$A_1$ and $A_2$ are as shown in Figure~\ref{fig12}.

\begin{figure}[h]
\centering
 \includegraphics[width=0.45\textwidth]{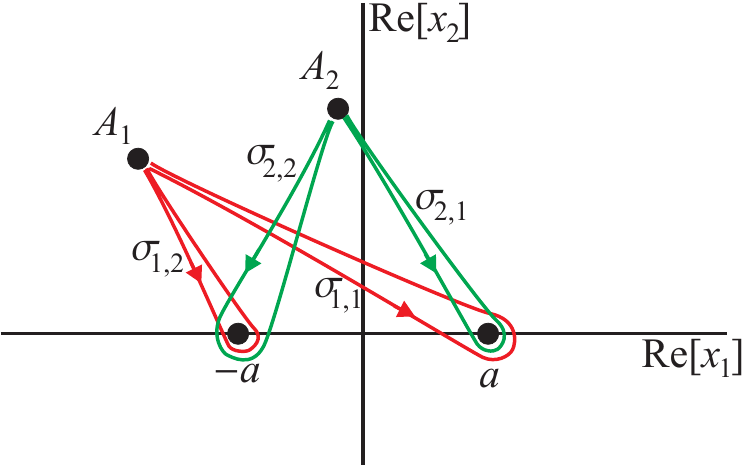}
  \caption{Trajectories of $A_{1,2}$ corresponding to the elementary paths $\sigma_{\ell,j}$}
\label{fig12}
\end{figure}

As it follows from the considerations in previous sections, each path 
$\sigma_{\ell,j}$
is an elementary bypass about the branch 2-line 
$L_\ell^{(j)}$.
The 2-lines  $L_1^{(j)}$ are bypassed in the positive direction, 
while the 
2-lines  $L_2^{(j)}$ are bypassed in the negative direction.

Now let us choose the basis functions $g_j$ of Theorem \ref{th:FiniteBasis}. 
For this, consider the contours $\Gamma_0$ and the double-eight
loops $\delta \Gamma_{\ell,j}$ for $\ell,j \in \{1,2\}$ shown in Figure~\ref{fig13}. 
For definiteness, we assume that the points marked by a small black circle on the contours belong to the physical sheet of~$S$. 

\begin{figure}[h]
\centering
 \includegraphics[width=0.9\textwidth]{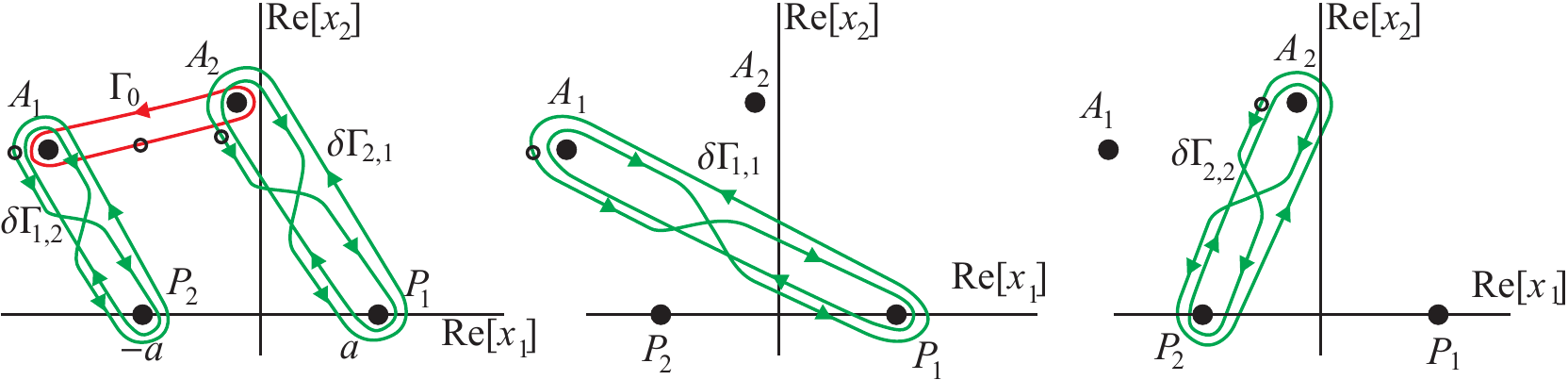}
  \caption{Basis contours for diffraction by a segment}
\label{fig13}
\end{figure}

Formally, since the branch points $P_{j}$ are of order~2, there should exist two double-eight loops for each pair of indexes $\ell,j$: $\delta \Gamma_{\ell,j}$ and $\delta \Gamma_{\ell,j}^{(1)}$. However, due to the identity (\ref{eq:ContourEq}), the loops $\delta \Gamma_{\ell,j}^{(1)}$ are not needed to be included into the basis, reducing the number of candidates for the basis functions from 9 to 5. 

Moreover, using contour deformation and cancellation, one can check the following contour identity
\begin{equation}
\delta \Gamma_{1,1} + \delta \Gamma_{2,1}
=
\delta \Gamma_{1,2} + \delta \Gamma_{2,2},
\label{eq:LinDep}
\end{equation}
by showing that each side of (\ref{eq:LinDep}) is equal to $2\Gamma_0$.
Therefore, 
one of the double-eight loops, for example $\delta \Gamma_{2,2}$, does 
not need to be included into the basis since it depends linearly on the other 3 contours. This effectively reduces the number of basis functions to 4, as required and these basis functions can be explicitly written as
\begin{align}
g_1(x_1, x_2) &= \int_{\Gamma_0} \left[ \dots \right] dl', 
& g_2(x_1, x_2) &= \int_{\delta \Gamma_{1,1}} \left[ \dots \right] dl', 
\label{eq:g12} \\
g_3(x_1, x_2) &= \int_{\delta \Gamma_{1,2}} \left[ \dots \right] dl', 
& g_4(x_1, x_2) &= \int_{\delta \Gamma_{2,1}} \left[ \dots \right] dl', 
\label{eq:g4}
\end{align}
where the  $\left[ \dots \right]$ integrand is the same as in (\ref{eq:ContFormula}).

Upon introducing the vector of functions
\begin{equation}
{\rm W} = (g_1 , g_2 , g_3 , g_4)^T, 
\label{eq:VecUn}
\end{equation} 
the analytic continuation of $u_{\rm c}$ is fully described by the following theorem.

\begin{theorem}
\label{th:Matrices}
The analytical continuation along 
each closed fundamental path $\sigma_{\ell,j}$ affects the vector of basis 
functions as follows: 
\begin{equation}
{\rm W} \stackrel{\sigma_{\ell,j}}{\longrightarrow} {\rm M}_{\ell,j} {\rm W},
\label{eq:Mat}
\end{equation}
where the $4\times4$ constant matrices $M_{\ell,j}$ are given by 
\begin{equation}
{\rm M}_{1,1} = \left( \begin{array}{cccc}
1 & -1 & 0 & 0 \\
0 & -1 & 0 & 0 \\
0 & -2 & 1 & 0 \\
0 & 0  & 0 & 1
\end{array} \right),
\qquad 
{\rm M}_{1,2} = \left( \begin{array}{cccc}
1 &  0 & -1 & 0 \\
0 &  1 & -2 & 0 \\
0 &  0 & -1 & 0 \\
0 &  0 &  0 & 1
\end{array} \right),
\label{eq:M1112}
\end{equation}
\begin{equation}
{\rm M}_{2,1} = \left( \begin{array}{cccc}
1 &  0 & 0 & -1 \\
0 &  1 & 0 & 0 \\
0 &  0 & 1 & 0 \\
0 & 0  & 0 & -1
\end{array} \right),
\qquad 
{\rm M}_{2,2} = \left( \begin{array}{cccc}
1 & -1 &  1 & -1 \\
0 &  1 &  0 & 0 \\
0 &  0 &  1 & 0 \\
0 & -2 &  2 & -1
\end{array} \right).
\label{eq:M2122}
\end{equation}
\end{theorem} 

The statement of the theorem can be checked directly by studying how the double-eight contours are transformed under the action of each $\sigma_{\ell,j}$, though it is omitted here for brevity.

Nevertheless, the following theorem and its proof can be considered as a check that the matrices given in Theorem \ref{th:Matrices} are indeed correct.

\begin{theorem}
\label{th:MatrixResults}
Let $j,k,\ell \in \{1,2\}$. The following statements are correct: 
\begin{itemize}

\item[a)] 
Each branch 2-line $L_\ell^{(j)}$ 
has order~2, i.e.\ :
\begin{equation}
u_{\rm c}(A ; (\sigma_{\ell,j})^2 ) = u_{\rm c}(A). 
\label{eq:Order}
\end{equation}

\item[b)]
The bypasses $\sigma_{1,j}$ and $\sigma_{2,k}$ commute for any pair $j,k$:
\begin{equation}
u_{\rm c}(A ; \sigma_{1,j} \sigma_{2,k} ) 
= 
u_{\rm c}(A ; \sigma_{2,k} \sigma_{1,j} ). 
\label{eq:Commute}
\end{equation}

\item[c)]
The intersecting branch 2-lines $L_1^{(j)}$ and $L_2^{(k)}$
have the additive crossing property: 
\begin{equation}
u_{\rm c}(A ) + u_{\rm c}(A ; \sigma_{1,j} \sigma_{2,k} ) 
=
u_{\rm c}(A ;  \sigma_{1,j} ) + u_{\rm c}(A ;\sigma_{2,k} ) .
\label{eq:Additive}
\end{equation} 

\end{itemize}

\end{theorem}

\begin{proof} 
To prove a) it is enough to show that for all $j,\ell\in\{1,2\}$, we have 
\[
{\rm M}_{\ell,j}^2={\rm I}_4,
\]
where ${\rm I}_4$ is the $4\times4$ identity matrix, while to prove b) one just needs to show that for all $j,k\in\{1,2\}$, we have
\[
{\rm M}_{1,j} {\rm M}_{2,k}   = {\rm M}_{2,k} {\rm M}_{1,j}.
\]
Finally, to prove c), it is enough to show that for all $j,k\in\{1,2\}$, we have
\[
{\rm I}_4 + 
{\rm M}_{1,j} {\rm M}_{2,k}   =  {\rm M}_{1,j} + {\rm M}_{2,k}.
\]
All this can be shown to be true directly for the matrices (\ref{eq:M1112}) and (\ref{eq:M2122}).
\end{proof}

Let us now discuss the proven relations. The relation (\ref{eq:Order}) can be 
considered as an alternative check of Theorem~\ref{th:Picard2}. 
Indeed, since $P_{1}$ and $P_2$ are branch points of order~2 on $S$, the branch 
2-lines $L_\ell^{(j)}$ have order~2. 

The second relation (\ref{eq:Commute}) follows from a fundamental property of
multidimensional complex analysis discussed at the end of Section \ref{sec:sec3analcont}:
the bypasses about $L_1^{(j)}$ and $L_2^{(k)}$ commute.

Finally, the relation (\ref{eq:Additive}) is discussed in details in~\cite{Assier2018a}.
This relation means that the function $u_{\rm c}$ 
can be represented locally as a sum of two functions 
having branch 2-lines, separately, at $L_1^{(j)}$ and at $L_2^{(k)}$. As shown in~\cite{Assier2018a} and \cite{Assier2019c} the additive 
crossing property plays a fundamental role
in the process of integration of functions of several complex variables.


\section{Link between finite basis and coordinate equations}
\label{sec:coordeq}

Here we continue to study the Dirichlet strip problem from the 
previous section. 

The property (\ref{eq:Mat}) of the vector ${\rm W}$ reminds of the behavior of a 
Fuchsian ordinary differential equation on a plane of a single
complex variable. Namely, 
the poles of the coefficients of a Fuchsian ODE are branch points of its solution, and the vector 
composed of linearly independent solutions is multiplied by a
constant  
{\em monodromy matrix\/} as the argument bypasses a branch point.  
Indeed, here, the ${\rm M}_{\ell,j}$ play the role of such monodromy matrices. 
It is also well known that, conversely, if a vector of linearly independent functions of a single 
variable have this behaviour, then there exists a Fuchsian equation obeyed by them (This is Hilbert's twenty-first problem). This can be shown using the concept of fundamental matrices and their determinants (the Wronskian), see \cite{Bolibrukh1992}.

Here the situation is more complicated. There are two independent complex variables
instead of one.
Moreover,  the behavior of the components of ${\rm W}$ at infinity are not explored. 
However, we can still prove some important statements. 

Throughout the paper, it is implicitly assumed that the field $u$, 
and, thus, its continuation basis 
${\rm W}$, depend 
on the angle of incidence $\vph^{\rm in}$. 
Let us now
consider four different incidence angles 
$\vph^{\rm in}_{\rm I}$,
$\vph^{\rm in}_{\rm II}$,
$\vph^{\rm in}_{\rm III}$,
and
$\vph^{\rm in}_{\rm IV}$. 
It hence leads to four
different wave fields, and to four basis vectors, all 
defined by (\ref{eq:VecUn}):
${\rm W}_{\rm I}$,
${\rm W}_{\rm II}$,
${\rm W}_{\rm III}$,
${\rm W}_{\rm IV}$.
Let us 
construct the $4 \times 4$ square matrix function ${\rm V}$ made of these
vectors, defined by
\begin{equation}
{\rm V} \equiv  \left( 
		{\rm W}_{\rm I}, 
        {\rm W}_{\rm II},
        {\rm W}_{\rm III},
        {\rm W}_{\rm IV},
  \right).
  \label{eq:MatV}
\end{equation}

We claim here (without proof) that the matrix ${\rm V}$ is non-singular
almost everywhere (in $\mathbb{C}^2$ minus a set having 
complex codimension~1), so that we can freely write ${\rm V}^{- 1}$. 

Note that the whole matrix is only branching at $T$, and that the equations 
(\ref{eq:Mat})
are valid for the matrix ${\rm V}$ as a whole: 
\begin{equation}
  {\rm V}  \overset{\sigma_{\ell,j}}{\longrightarrow}  {\rm M}_{\ell ,j}
  {\rm V}. 
  \label{eq:MatV1}  
\end{equation}

This allows us to formulate the following theorem, linking the theory of differential equations
to the strip diffraction problem:

\begin{theorem}
  \label{th:thfrobetal}
	There exist two $4 \times 4$ matrix functions ${\rm Z_1} (x_1 , x_2)$ and
	${\rm Z_2} (x_1 , x_2)$, meromorphic in $\mathbb{C}^2$, such that    
    \begin{itemize}
    \item[a)]
     the matrix function ${\rm V}$ satisfies the following differential equations
    \begin{equation}
      \ptl_{x_1} {\rm V}= {\rm V}\, {\rm Z}_1  
      \quad \text{ and } \quad 
      \ptl_{x_2} {\rm V}= {\rm V}\, {\rm Z}_2; 
      \label{eq:coordequations}
    \end{equation}
    
    \item[b)] 
    these matrix functions obey the consistency relation:
    \begin{equation}
      {\rm Z}_1\, {\rm Z}_2 -{\rm Z}_2\, {\rm Z}_1  
      = 
       \ptl_{x_2} {\rm Z}_1 - \ptl_{x_1} {\rm Z}_2 ,
    \label{eq:consistency}
    \end{equation}
   \end{itemize}
  where $\ptl x_\ell$ for $\ell \in \{ 1, 2 \}$ are the complex
  derivatives defined in (\ref{eq:ComplexHelmholtz}).
\end{theorem}

\begin{proof}
For a)  assume that ${\rm V}$ is known and that
${\rm V}^{- 1}$ exists almost everywhere (in $\mathbb{C}^2$ minus a set of complex 
codimension~1). In this case, the
coefficients ${\rm Z}_1$ and ${\rm Z}_2$ are simply given by
  \begin{equation}
    {\rm Z}_1 ={\rm V}^{- 1} \ptl_{x_1} {\rm V} 
       \quad \text{ and } \quad 
    {\rm Z}_2 ={\rm V}^{- 1} \ptl_{x_2} {\rm V}.
  \end{equation}
Let us show that the matrices ${\rm Z}_1$ and ${\rm Z}_2$ are single-valued
in $\mathbb{C}^2$. The only sets at which one can expect branching are 
the branch 2-lines of ${\rm V}$, i.e.\ $L_\ell^{(j)}$. Make a bypass 
$\sigma_{\ell,j}$ about a 2-line $L_\ell^{(j)}$ and study the change of ${\rm Z}_1$ 
and ${\rm Z}_2$ as the result of this bypass:
\begin{equation}
{\rm Z}_{k} \stackrel{\sigma_{\ell,j}}{\longrightarrow}    
{\rm V}^{-1} {\rm M}_{\ell,j}^{-1} \ptl_{x_k}\left( {\rm M}_{\ell,j} {\rm V}\right) = {\rm Z}_k  
,
\quad k \in \{1,2 \} ,
\end{equation}  
because ${\rm M}_{\ell,j}$ are constant matrices.
Thus, the coefficients ${\rm Z}_1$ and ${\rm Z}_2$ are not changing at the branch 2-lines of 
${\rm V}$, and, therefore, they are single-valued in $\mathbb{C}^2$. A detailed study 
shows that they have simple polar sets at the lines $L_{\ell}^{(j)}$, and, possibly, polar sets at the zeros of ${\rm det}({\rm V})$ though we omit this discussion here for brevity.  
  
To prove b), differentiate the first equation of (\ref{eq:coordequations})
with respect to $x_2$, and the second equation with respect to $x_1$. 
The expressions in the left are equal, and we get 
\[
\left(\ptl_{x_2} {\rm V}\right)  {\rm Z}_1 + {\rm V} \ptl_{x_2} {\rm Z}_1
= 
\left(\ptl_{x_1} {\rm V}\right) {\rm Z}_2 + {\rm V} \ptl_{x_1} {\rm Z}_2.
\]
Applying (\ref{eq:coordequations})
and multiplying by ${\rm V}^{-1}$, obtain (\ref{eq:consistency}).

As it follows from Frobenius theorem, this relation guarantees the solvability 
of the system (\ref{eq:coordequations}). 
\end{proof}

A detailed form of the coefficients ${\rm Z}_1$ and ${\rm Z}_2$ can be found 
in \cite{Shanin2002a,Shanin2008,Shanin2008a}. Here our aim was just to demonstrate 
that the existence of the coordinate equations is connected with the structure of analytical continuation of the solution. 

Finally, before concluding the paper, it is interesting to note that the idea of considering a set of different incident angles
and trying to link the solutions to each other by means of differential equations is
somewhat reminiscent of Biggs' interpretation of embedding formulae
{\cite{Biggs2006}}.


\section{Conclusion}

We have provided an explicit method to analytically continue two-dimensional wave fields
emanating from a broad range of diffraction problems and
described the singular sets (in $\mathbb{C}^2$) of their analytical continuation. We have shown that, even though the analytical continuation may have potentially infinitely many
branches, each branch can be expressed as a linear combination
of finitely many basis functions. Such basis functions are expressed as
Green's integrals over a real double-eight contour. The effectiveness of the general theory was illustrated via the example of diffraction by an ideal strip, for which we proved that only 4 basis functions were needed. Using these, we were able to completely describe the analytical continuation and study its branching behaviour. Finally, we have shown that this finite basis property was directly related to the existence of the so-called coordinate equation for the strip problem.

\section*{Acknowledgement}

	R.C. Assier would like to acknowledge the support by UK EPSRC
	(EP/N013719/1). Both authors would like to thank the Isaac Newton Institute for
	Mathematical Sciences, Cambridge, for support and hospitality during the programme
	``Bringing pure and applied analysis together via the Wiener--Hopf technique, its
	generalisations and applications'' where some work on this paper was undertaken. This work was supported by EPSRC (EP/R014604/1) and, in the case of A.V. Shanin, the Simons foundation. Both authors are also grateful to the Manchester Institute for Mathematical Sciences for its financial support.



\clearpage
\bibliographystyle{RS}
\bibliography{biblio}

\appendix
\numberwithin{equation}{section}

\section{Diffraction problems on a plane and on Sommerfeld surfaces}
\label{app:generalclass}
In this appendix, we aim to describe the wide class of 2D diffraction problems for which the theory developed in the paper is valid.

Consider an incident plane wave $u_\text{in}$ impinging on a set of scatterers. The aim is to find the resulting total field $u$ satisfying the Helmholtz equation and subject to some specified boundary conditions on the scatterers. For the problem to be well-posed, we also require that $u$ have bounded energy at the edges of the scatterers (Meixner condition) and that the scattered field $u-u_\text{in}$ be outgoing (radiation condition).

For the theory of the paper to be applicable to such problem, it should obey the four following simple rules:

\vspace*{-0.37cm}
\begin{enumerate}[label={\color{blue}\textup{R\arabic*}}]
	\item \label{item:rule1}all scatterers faces should be straight segments, possibly intersecting, possibly of infinite length;
	\item \label{item:rule2} the boundary conditions should be of Neumann or
	Dirichlet type, they are allowed to be different on each side of
	the segment;
	\item \label{item:rule3} the segments should either have the same supporting line or, failing that, their supporting lines should all intersect in one common point;
	\item \label{item:rule4}the angles between the supporting lines of each segment should be rational
	multiples of $\pi$.
\end{enumerate}
\vspace*{-0.37cm}

Provided that \ref{item:rule1}--\ref{item:rule4} are satisfied, the diffraction problem can be reformulated as a propagation problem on a Sommerfeld surface, without scatterers \textit{per se}, but with a finite number $N$ of branch points and a finite number $M$ of sheets.

One can formalise the procedure of building the Sommerfeld surface $S$ and the solution $u$ on it provided that the solution $u$ in the physical domain is known. Consider the origin of the polar coordinates $(\rho, \varphi)$ 
	to be the intersection point mentioned in the rule~\ref{item:rule3} (or anywhere on the support line if there is only one). The solution $u$ can be continued past the faces of the scatterers by following the reflection equations across each face of the scatterers:    
\begin{eqnarray}
\text{(Dirichlet):} \quad & u (\rho, \Phi + \varphi) = -
u (\rho, \Phi - \varphi), & \\
\text{(Neumann):} \quad & u (\rho, \Phi + \varphi) = +
u (\rho, \Phi - \varphi) , & 
\end{eqnarray}
	where $(\rho,\Phi)$ is a point belonging to a scatterer's face.
	
	In other words, one should take the physical domain with the solution $u$ on it, make cuts along each scatterers face, make enough copies of the physical plane by reflection across some lines (the supporting lines of the scatterers and of their successive reflections), and attach the obtained reflected planes to each other according to these reflection equations.

{\textbf{Examples.}} As illustrated in Figure \ref{fig:differentproblemsN1},
examples satisfying the conditions \ref{item:rule1}--\ref{item:rule4} with $N
= 1$ branch point include: the Dirichlet-Dirichlet (Dir-Dir) or
Neumann-Neumann (Neu-Neu) half-plane ($M= 2$), the
Dirichlet-Neumann (Dir-Neu) half-plane ($M= 4$) and any Dir-Dir or
Neu-Neu wedge with internal angle $p \pi / q$ with $(p, q) \in \mathbb{N}_{>
	0}$ and $q > p$, for which $M$ is equal to
the denominator of the irreducible form of the fraction $q / (2 q - p)$. All
these can be solved by the Wiener-Hopf technique, see e.g. {\cite{Noble1958}}
for Figure \ref{fig:differentproblemsN1}{\color{myred}a} , {\cite{Rawlins1975}} for Figure
\ref{fig:differentproblemsN1}{\color{myred}d}  and {\cite{Nethercote2020}} for Figure
\ref{fig:differentproblemsN1}{\color{myred}b}  and \ref{fig:differentproblemsN1}{\color{myred}d} .

\begin{figure}[h]
	\centering
	\includegraphics[width=0.99\textwidth]{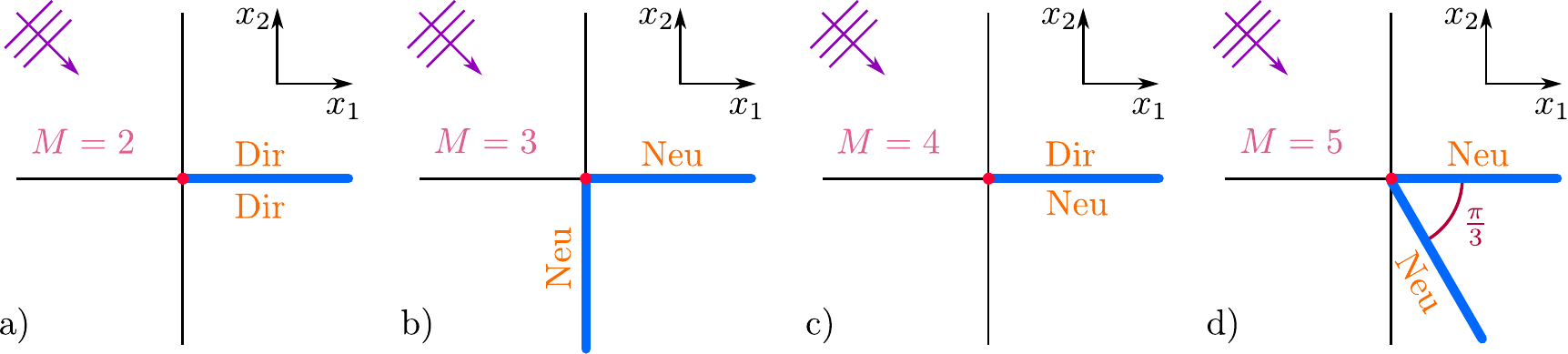}
	\caption{Example of scatterers (thick blue lines) with $N = 1$ branch point
		(thick red dot). }
	\label{fig:differentproblemsN1}
\end{figure}

Examples with higher-numbers of branch points, illustrated in Figure
\ref{fig:exNbig1}, include the Dir-Dir or Neu-Neu strip ($N = 2$,
$M= 2$), the Dir-Neu strip $(N = 2, M= 4)$, or more
complicated scatterers such as two intersecting Dir-Dir strips with respective
angle $\pi / 4$ and $3 \pi / 4$, for which branch points not belonging to the
physical scatterers start to occur $(N = 9, M= 8)$, etc.

\begin{figure}[h]
	\centering
	\includegraphics[width=0.75\textwidth]{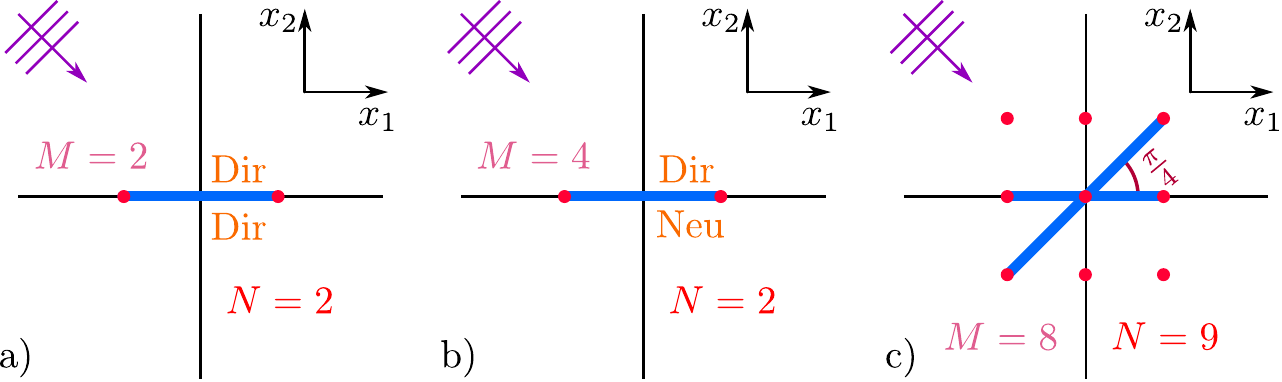}
	\caption{Example of scatterers (thick blue lines) with $N > 1$ branch points
		(thick red dots)}
	\label{fig:exNbig1}
\end{figure}

Finally, configurations with multiple non-intersecting scatterers are also possible, as illustrated in Figure \ref{fig:exnonintersectingmultiple}. The case of multiple aligned strips (Figure~\ref{fig:exnonintersectingmultiple}{\color{myred}a}) is a diffraction problem that has previously been investigated in e.g.\ \cite{Shanin2002a} (via coordinate equations) or \cite{Priddin2019} (via an iterative Wiener--Hopf method). 

\begin{figure}[h]
	\centering
	\includegraphics[width=0.5\textwidth]{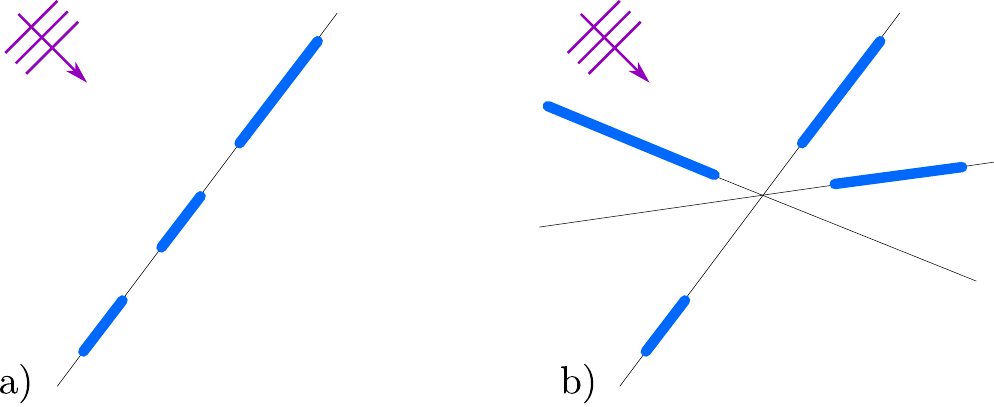}
	\caption{Example of multiple non-intersecting scatterers (thick blue lines), the supporting lines are illustrated in thin black lines}
	\label{fig:exnonintersectingmultiple}
\end{figure}

\section{Complex Green's theorem}
\label{app:Greens}

In this appendix we aim to formulate a complex version of Green's theorem. In order to do so we need to introduce the notion of \textit{complex gradient} of a function of several complex variables.
Because of the topic of the paper, it is enough to focus on $\mathbb{C}^2$ for which it can be defined as follows. For a function $v(x_1,x_2)$ of two complex variables $x_1$ and $x_2$, the complex gradient of $v$, denoted $\nabla_C v$, is defined as the complex 1-form
\[ \nabla_C v \equiv \partial_{x_1}\!v \, dx_2 - \partial_{x_2}\!v\, dx_1 . \]
The complex Green's theorem in $\mathbb{C}^2$ can hence be formulated as follows. 

\begin{theorem}[Complex Green's theorem]
	\label{th:complexgreenth}
	If two functions $v (x_1,
	x_2)$ and $w (x_1, x_2)$ both obey the complex Helmholtz equation (\ref{eq:ComplexHelmholtz}) in some
	neighbourhood (included in $\mathbb{C}^2$) where they are analytic, i.e.\ where they satisfy the Cauchy-Riemann conditions (\ref{eq:Cauchy}), then, in this neighbourhood,
	\[ d (v \nabla_C w - w \nabla_C v) = 0, \]
	where $d$ is the usual exterior derivative operator for differential forms.
\end{theorem}

\begin{proof}
	Following {\cite{Shabat2}}, it is convenient to decompose the
	exterior derivative $d$ as $d = \partial + \bar{\partial}$, where $\partial$ and $\bar{\partial}$ are the so-called Dolbeault operators. For any
	function (0-form) $f (x_1, x_2)$, they are defined by
	\begin{align*}
		\partial f = \partial_{x_1} f dx_1 + \partial_{x_2} f dx_2 \quad
		\text{ and } \quad \bar{\partial} f = \partial_{\bar{x}_1} f d\bar{x}_1 +
		\partial_{\bar{x}_2} f d \bar{x}_2,
	\end{align*}
	while for any complex 1-form $\omega = f_1 (x_1, x_2) d x_1 + f_2 (x_1, x_2)
	d x_2 + g_1 (x_1, x_2) d \bar{x}_1 + g_2 (x_1, x_2) d \bar{x}_2$, they are given by
	\begin{align*}
		\partial \omega & =  \partial f_1 \wedge d x_1 + \partial f_2 \wedge
		d x_2 + \partial g_1 \wedge d \bar{x}_1 + \partial g_2 \wedge
		d \bar{x}_2,\\
		\bar{\partial} \omega & =  \bar{\partial} f_1 \wedge d x_1 +
		\bar{\partial} f_2 \wedge d x_2 + \bar{\partial} g_1 \wedge d
		\bar{x}_1 + \bar{\partial} g_2 \wedge d \bar{x}_2 .
	\end{align*}
	Hence, using the fact that $v$, $w$ and their derivatives are analytic we can
	show that \\ $\bar{\partial} (v \nabla_C w - w \nabla_C v) = 0$, leading to
	\begin{align*}
		d (v \nabla_C w - w \nabla_C v) & =  \partial (v \nabla_C w - w
		\nabla_C v)\\
		& =  (\partial_{x_1} (v \partial_{x_1} w) + \partial_{x_2} (v
		\partial_{x_2} w) - \partial_{x_1} (w \partial_{x_1} v) - \partial_{x_2} (w
		\partial_{x_2} v)) d x_1 \wedge d x_2,
	\end{align*}
	where we have used that $d x_{\ell} \wedge d x_{\ell} = 0$ for $\ell
	\in \{ 1, 2 \}$ and that $d x_1 \wedge d x_2 = - d x_2 \wedge
	d x_1$. 
	
	Upon expanding out the derivatives, we obtain
	\begin{align*}
		d (v \nabla_C w - w \nabla_C v) =  (v (\partial_{x_1}^2 w +
		\partial_{x_2}^2 w) - w (\partial_{x_1}^2 v + \partial_{x_2}^2 v)) d
		x_1 \wedge d x_2 = 0,
	\end{align*}
	as required, since both $v$ and $w$ satisfy the complex Helmholtz equation.
\end{proof}

	Note that the same result can be proven in $\mathbb{C}^3$. In order to do so, one has to define the complex gradient of a function $v(x_1,x_2,x_3)$ as
	\[ \nabla_C v = \partial_{x_1}v\, d x_2 \wedge d x_3 +
	\partial_{x_2}v \, d x_3 \wedge d x_1 + \partial_{x_3}v\, d x_1 \wedge d x_2 . \]
	These complex gradients definitions are closely related to and can be expressed in terms of the Hodge star operator.



The complex Green's theorem, combined with Stokes' theorem for complex differential forms, makes it possible to show that for any two functions $v$ and $w$ satisfying the hypotheses of Theorem \ref{th:complexgreenth}, and two contours $\gamma_1$ and $\gamma_2$ that can be deformed homotopically to each other within the region of analyticity of $v$ and $w$, we have
\begin{align}
\int_{\gamma_1} \left( v \nabla_C w - w \nabla_C v \right) = \int_{\gamma_2} \left( v \nabla_C w - w \nabla_C v\right). \label{eq:GreencontourdefappC}
\end{align}
Hence, the value of the integral is not changed when the contour is deformed, even within $\mathbb{C}^2$. Moreover, note that when restricted to a contour $\gamma$ in $\mathbb{R}^2$, we have
\begin{align*}
v \nabla_C w - w \nabla_C v=\left[ v \frac{\partial w}{\partial n}-w \frac{\partial v}{\partial n} \right] dl.
\end{align*}
The link with what we have done in the paper becomes clear by choosing $v\equiv u_c$ and $w\equiv G$.

To prove that $u_c(A)$ is indeed defined uniquely by (\ref{eq:initialgreensidentity}) in a small complex neighbourhood of $A_0$, proceed as follows\footnote{Below, square brackets following either $\mathbb{C}^2$ or $\mathbb{R}^2$ are used to specify the coordinate system under consideration.}. Let $A\equiv(x_1^A,x_2^A)\in\mathbb{C}^2[x_1,x_2]\setminus T$ in a small neighbourhood of $A_0\in S$, and let $u_c(A)$ be the value obtained by letting ${\bf r}$ become a complex vector, and choosing $A$ close enough to $A_0$ such that $G({\bf r},{\bf r}')$ remains regular for ${\bf r}'\in \gamma$. 
	
	Consider now the change of variable $(\xi_1,\xi_2)=(x_1-x_1^A,x_2-x_2^A)$ so that in $\mathbb{C}^2[\xi_1,\xi_2]$, we have $A\equiv(0,0)$, so that $A\in\mathbb{C}^2[\xi_1,\xi_2]\cap\mathbb{R}^2[\xi_1,\xi_2]$, and $u_c$ can be studied as a function obeying the Helmholtz equation on the real plane $\mathbb{R}^2[\xi_1,\xi_2]$. This can be done by considering the function $\tilde{u}_c(\xi_1,\xi_2)=u_c(\xi_1+x_1^A,\xi_2+x_2^A)$, and, as such, it is given by the Green's formula
\begin{align*}
u_c(A)=\tilde{u}_c(0,0) &=\int_{\tilde{\gamma}} \left[ \tilde{u}_c(\tilde{{\bf r}}') \frac{\partial G}{\partial n'}({\bf 0},\tilde{{\bf r}}')-G({\bf 0},\tilde{{\bf r}}') \frac{\partial \tilde{u}_c}{\partial n'}(\tilde{{\bf r}}') \right] d\tilde{l}' = \int_{\tilde{\gamma}} \left( u_c \nabla_C G - G \nabla_C u_c \right),
\end{align*}
where $\tilde{\gamma}$ is a real contour encircling the origin of $\mathbb{R}^2[\xi_1,\xi_2]$, $\tilde{{\bf r}}'$ is a real vector pointing to a point in $\tilde{\gamma}$ and ${\bf 0}$ is the zero vector. The last equality comes from the fact that a form is independent of the coordinate system in which it is expressed, and from the translational invariance of $G$ that satisfies $G(x_1-x_1^A,x_2-x_2^A;x_1'-x_1^A,x_2'-x_2^A)=G(x_1,x_2;x_1',x_2')$. 

Now when viewed in $\mathbb{C}^2[x_1,x_2]$, $\tilde{\gamma}$ is not a contour in $\mathbb{R}^2[x_1,x_2]$, but by (\ref{eq:GreencontourdefappC}) we can deform this contour to the contour $\gamma\subset\mathbb{R}^2[x_1,x_2]$ used in (\ref{eq:initialgreensidentity}) without changing the value of the integral to get
\begin{align*}
u_c(A) &= \int_{\gamma} \left( u_c \nabla_C G - G \nabla_C u_c \right)=\int_{\gamma} \left[ u_c({\bf r}') \frac{\partial G}{\partial n'}({\bf r},{\bf r}')-G({\bf r},{\bf r}') \frac{\partial u_c}{\partial n'}({\bf r}') \right] dl',
\end{align*}
as expected, where ${\bf r'}$ is a real vector pointing to $\gamma$ and ${\bf r}$ is a complex vector pointing to $A$, which is exactly what we would have obtained by letting ${\bf r}$ become complex in (\ref{eq:initialgreensidentity}). The contours used and the subsequent deformation are illustrated in Figure \ref{fig:appCfig}.
\begin{figure}[h]
	\centering
	\includegraphics[width=0.95\textwidth]{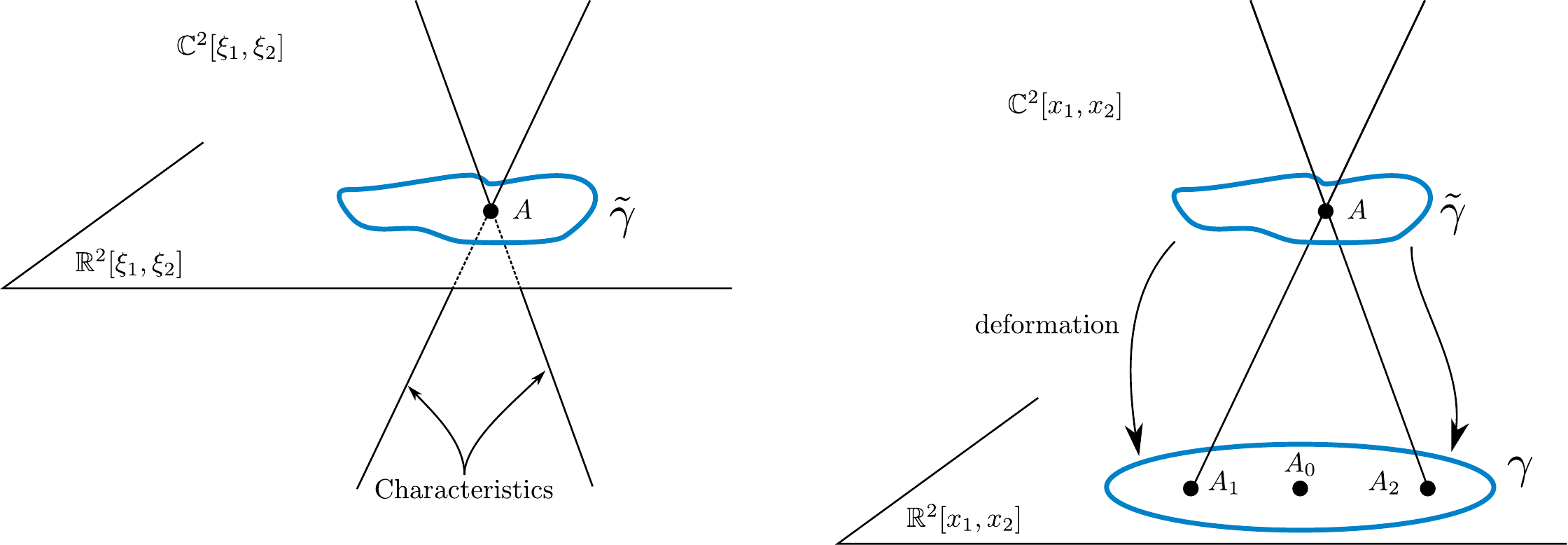}
	\caption{Diagrammatic Illustration of the contours $\tilde{\gamma}$ and $\gamma$ and how one is deformed to the other}
	\label{fig:appCfig}
\end{figure}

\section{Bypasses and argument of the Hankel function}
\label{app:bypasshankel}

Let $A = (x_1, x_2) \in \mathbb{C}^2 \setminus (T \cup \mathbb{R}^2)$. We are interested in the
behaviour of the function
\begin{align*}
	\mathcal{G} (x_1', x_2') & =  H_0^{(1)} \left( \mathscr{k} \sqrt{(x_1 - x_1')^2 +
		(x_2 - x_2')^2} \right),
\end{align*}
in the real plane $(x_1', x_2') \in \mathbb{R}^2$. In particular we want to
know how this function changes along a given contour $\Gamma$ within this real
plane. Upon introducing the two fixed complex quantities $z_1 = x_1 + i x_2$
and $z_2 = x_1 - i x_2$ we can rewrite $\mathcal{G}$ as
\begin{align*}
	\mathcal{G} (x_1', x_2') & = H_0^{(1)} \left( \mathscr{k} \sqrt{[z_1 - (x_1' + i
		x_2')] [z_2 - (x_1' - i x_2')]} \right) .
\end{align*}
For this purpose, it is also convenient to see the real $(x_1', x_2')$ real
plane as a complex plane of the one complex variable $z' = x_1' + ix_2'$, so
that $\mathcal{G}$ can be considered as a function of the complex variable $z'$ and we
may write
\begin{align*}
	\mathcal{G} (z') 
	& =  H_0^{(1)} \left( \mathscr{k} \sqrt{[z' - z_1] \overline{[z' -
		\overline{z_2}]}} \right).
\end{align*}
In the $z'$ complex plane, the
points $A_1$ and $A_2$ are defined as the points $z' = z_1$ and $z' =
\overline{z_2}$ respectively.

Consider now a contour fragment $\Gamma$ that encircles the point $A_1$, i.e. the
point $z' = z_1$, once in the anticlockwise direction. It does not encircle
the point $A_2$. Upon translating the origin of the $z'$ plane by considering
the new variable $\xi'_1 = z' - z_1$, we can rewrite $G$ in the $\xi'_1$
plane
as
\begin{align*}
	\mathcal{G} (\xi'_1) & =  H_0^{(1)} \left( \mathscr{k} \sqrt{\xi'_1 \overline{[\xi'_1 + z_1 -
		\overline{z_2}]}} \right) .
\end{align*}
In the $\xi'_1$ plane, $\Gamma$ encircles $0$ once anticlockwise, but does not
encircle $- z_1 + \overline{z_2}$, so that if we parametrise $\Gamma$ by $\tau
\in [0, 1]$, we have
\begin{align*}
	\xi'_1 (0) & \overset{\Gamma}{\rightarrow} \xi'_1 (1) = e^{2 i \pi} \xi_1' (0),\\
	\overline{[\xi'_1 (0) + z_1 - \overline{z_2}]} & \overset{\Gamma}{\rightarrow} 
	 	\overline{[\xi'_1 (1) + z_1 - \overline{z_2}]}= \overline{[\xi'_1 (0) + z_1 - \overline{z_2}]}.
\end{align*}
Hence, if one introduces the argument of the Hankel function to be
\begin{align*}
	\mathcal{Z} & =   \mathscr{k} \sqrt{\xi'_1 \overline{[\xi'_1 + z_1 -
			\overline{z_2}]}},
\end{align*}
we have
\begin{align*}
	\mathcal{Z} (0) = \mathscr{k} \sqrt{\xi'_1(0) \overline{[\xi'_1(0) + z_1 -
			\overline{z_2}]}} & \overset{\Gamma}{\rightarrow}  \mathscr{k}
	\sqrt{e^{2 i \pi} \xi'_1(0) \overline{[\xi'_1(0) + z_1 -
			\overline{z_2}]}} =
	e^{i \pi} \mathcal{Z} (0).
\end{align*}
Hence, when a fragment of a contour encircles $A_1$ once in the anticlockwise
direction, the argument of the Hankel function is indeed changed from
$\mathcal{Z} \rightarrow e^{i \pi} \mathcal{Z}$. Naturally if the bypass was
clockwise, the argument would pass from $\mathcal{Z} \rightarrow e^{- i \pi}
\mathcal{Z}$.

A similar argument, introducing the new variable $\xi_2' = z' - \overline{z_2}$,
leads to the fact that when a fragment of a contour encircles $A_2$ once in
the anticlockwise direction, the argument of the Hankel function is indeed
changed from $\mathcal{Z} \rightarrow e^{- i \pi} \mathcal{Z}$. Naturally if
the bypass was clockwise, the argument would pass from $\mathcal{Z}
\rightarrow e^{i \pi} \mathcal{Z}$.

\end{document}